\documentclass[reprint,superscriptaddress,nofootinbib, amsmath,amssymb, groupedaddress,aps,pra]{revtex4-2}

\usepackage{graphicx}% Include figure files
\usepackage{dcolumn}% Align table columns on decimal point
\usepackage{bm}% bold math
\usepackage{braket}

\usepackage{amsthm}
\usepackage{hyperref}% add hypertext capabilities
\hypersetup{
    colorlinks=true,
    linkcolor=blue,
    citecolor=red,
    filecolor=magenta,      
    urlcolor=blue,
}
\usepackage{cleveref}% Enhanced referencing

\newtheorem{definition}{Definition}
\newtheorem{theorem}{Theorem}

\newtheorem{problem}{Problem}

\newtheorem{model}{Model}

\makeatletter
\def\l@subsubsection#1#2{}
\makeatother

\begin{document}

\title{Spectral methods: crucial for machine learning, natural for quantum computers?}
\author{Vasilis Belis}
\author{Joseph Bowles}
\author{Rishabh Gupta}
\author{Evan Peters}
\author{Maria Schuld}\thanks{In alphabetical order}
\affiliation{Xanadu Quantum Technologies Inc., Toronto, Ontario, M5G 2C8, Canada}

\date{\today} 

\begin{abstract}
This article presents an argument for why quantum computers could unlock new methods for machine learning. We argue that spectral methods, in particular those that learn, regularise, or otherwise manipulate the Fourier spectrum of a machine learning model, are often natural for quantum computers. For example, if a generative machine learning model is represented by a quantum state, the Quantum Fourier Transform allows us to manipulate the Fourier spectrum of the state using the entire toolbox of quantum routines, an operation that is usually prohibitive for classical models. At the same time, spectral methods are surprisingly fundamental to machine learning: A spectral bias has recently been hypothesised to be the core principle behind the success of deep learning; support vector machines have been known for decades to regularise in Fourier space, and convolutional neural nets build filters in the Fourier space of images. Could, then, quantum computing open fundamentally different, much more direct and resource-efficient ways to design the spectral properties of a model? We discuss this potential in detail here, hoping to stimulate a direction in quantum machine learning research that puts the question of ``why quantum?'' first.   
\end{abstract}

\maketitle

\section{Introduction}

In the search for real-world applications for quantum computers, there is a growing consensus that cryptoanalysis and quantum simulation are the most mature proposals at this stage, and finding others ``has proven remarkably difficult'' \cite{babbush2025grand}. Both are also somewhat exceptional: cryptography is arguably the most structured of all real-world problems, and quantum chemistry shares its very theoretical foundations with quantum computers. In this perspective we propose a motivation for an application area where a clear-cut bridge has been elusive so far \cite{schuld2022quantum}: machine learning. 

Besides attracting excessive attention in academic and industrial research, quantum machine learning \cite{cerezo2022challenges, schuld2021machine, biamonte2017quantum} has struggled to make a clear case for why quantum computing should be of fundamental interest for generalising from data. For example, speeding up linear algebra \cite{rebentrost2014quantum, kerenidis2016quantum, wiebe2012quantum, lloyd2014quantum} fails to convincingly account for the time it takes to load big matrices into a realistic quantum computer; ``Quantum Neural Networks'' \cite{farhi2018classification, schuld2020circuit} still need to prove that they deserve their spot in the gallery of scalable, performant models \cite{cerezo2025does, bowles2024better, recio2025train}, and reverting to the analysis of ``quantum data'' \cite{huang2025generative, huang2022quantum} needs to make a case for why leaving states unmeasured is critical to learn their properties in practically relevant situations. Why, then, should we believe that quantum computers are useful for machine learning?

\begin{figure}[t]
    \centering
    \includegraphics[width=0.9\linewidth]{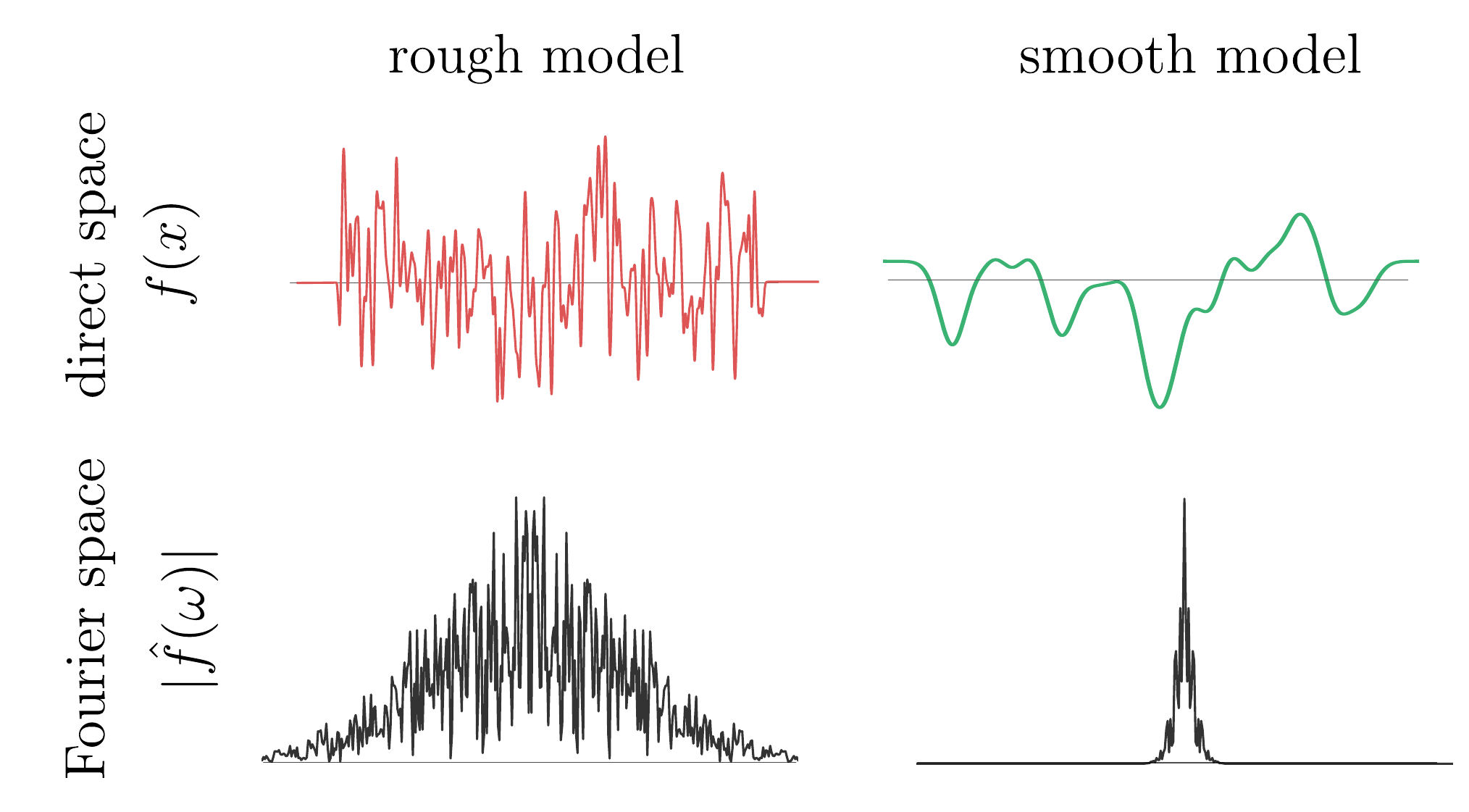}
    \caption{A common simplicity bias of machine learning models is their smoothness, which is linked to a decay of the model function's Fourier spectrum. But designing such a decay in Fourier space is usually computationally expensive. Can quantum computers help here?   }
    \label{fig:hero}
\end{figure}

Here we discuss a potential answer to this question. It rests on the observation that ``spectral methods''---which we define loosely as the design of machine learning models with desirable properties in Fourier space---are fundamental to a core principle of learning \cite{wilson2025deep}: to find simple models in expressive model classes. At the same time, quantum algorithms often rely on, or can be understood by, Fourier analysis. The main point of this paper is to explain why this connection could be a fruitful starting point for the design of quantum machine learning algorithms.

As an illustrative example for the importance of spectral methods in machine learning, consider a probability distribution $p(x)$ that tells us how likely it is to sample a data point $x\in \mathbb{R}^N$. Simple models intuitively correspond to \textit{smooth} distributions $p(x)$, where we use ``smooth'' broadly to mean Lipschitz and infinitely differentiable functions. Such models are robust to input perturbations: when we change the input a little, the probability $p(x + \delta)$ should not change a lot. Furthermore, smooth models have special properties in Fourier space: their spectrum decays super-polynomially (see Figure \ref{fig:hero}). This can be intuitively seen by considering the Fourier decomposition of $p(x)$,
\begin{align}\label{eq:p_fourier}
    p(x) = \frac{1}{\sqrt{2 \pi}}\int\limits_{\Omega}  \;\; \hat{p}(\omega) e^{-2 \pi i x \omega } d \omega,
\end{align}
where $\hat{p}(\omega)$ are the Fourier coefficients and  $\omega \in \Omega \subset \mathbb{R}^N$ the \textit{frequencies}. Functions that are smooth will be composed of weakly oscillating basis functions $e^{-2 \pi i x \omega }$, and therefore have small coefficients $\hat{p}(\omega)$ for large frequencies $|\omega|$ in Eq.~\ref{eq:p_fourier}. 

There has been growing recognition in the machine learning community that smoothness is crucial for a model to learn and generalise from data. To this end, an abundance of techniques have been developed for imposing smoothness \cite{JMLR:v15:srivastava14a,noh2017regularizing,miyato2018spectral,dherin2022why}. At the same time, these heuristic techniques are indirect and are often computationally inefficient \cite{pmlr-v137-rosca20a}, and so imposing model smoothness has remained out of reach of these classical methods. As shown above, the important concept of a simplicity bias in learning translates to a clearly defined behaviour in Fourier space that can help design models, an idea with a powerful group-theoretic generalisation that we will deal with in depth here \cite{diaconis1988group, kondor2008group}.  

Given that regularisation, or biasing a machine learning method towards simple models, is one of the most fundamental themes in machine learning, it is curious that spectral methods enforcing smoothness in Fourier space are not more prominent in the canonical knowledge base of researchers (and, by extension, scarce in the literature of quantum machine learning). A reason might be the computational challenge of working in the Fourier space of a large model, which is often only accessible indirectly via the convolution theorem. This theorem states that changing the Fourier coefficients $\hat{f}$ of a model by multiplying it with a filter $\hat{g}$ in Fourier space corresponds to a convolution in direct space:
\begin{align}
    (f * g)(x) = \int_X f(y) g(x-y) dy = \mathcal{F}^{-1}(\hat{f} \cdot \hat{g}).
\end{align}
In machine learning, the function $g$---whose Fourier coefficients $\hat{g}$ make up the filter---is usually known as a (stationary) \textit{kernel}. As a consequence, a wide range of kernel methods, which were the workhorse of machine learning in the 1990's and early 2000's \cite{kernel, steinwart2008support} and are widely used for small-to-medium data problems, can be thought of as spectral regularisation methods. The same is true for kernel tools like the Maximum Mean Discrepancy employed to train implicit generative models \cite{gretton2012kernel, li2017mmd}. Likewise, convolutional layers of neural networks, and their group-theoretic generalisations \cite{bronstein2021geometric, kondor2018generalization} such as graph neural networks \cite{wu2020comprehensive}, design a model class by shaping the  Fourier spectrum with a filter (although this filter acts on the image rather than the model function, which is computationally much easier). Last but not least, the convolution theorem helps to empirically probe how deep learning works: The \textit{spectral bias} hypothesis \cite{rahaman2019spectral, xu2019frequency} states that the complex model class of vastly overparametrised neural networks learns to match the ground truth's frequency from smallest (i.e., the ``smoothest'' component) to largest (the most oscillating component). All these observations point towards the fact that the Fourier spectrum of a model is a crucial mathematical object to study and design good machine learning models, but that we have to rely on indirect access to Fourier space due to computational challenges. 

Let us now consider a machine learning model that is represented by a trainable \cite{mitarai2018quantum, schuld2019evaluating} quantum state $\ket{\psi_{\theta}}$. How can quantum computers possibly help to access or shape the model's Fourier spectrum more directly? As we will see, there are many possible answers to this question, which intimately depend on how the model is defined, and how its spectrum is supposed to be designed. For example, if we consider the probability of a measurement outcome, $p(x) = |\langle x |\psi_{\theta}\rangle|^2$, as an implicit generative model \cite{liu2018differentiable, rudolph2024trainability} we can use the \textit{Quantum Fourier Transform} (QFT) \cite{moore2006generic} to apply a Fourier transform to the amplitudes of the quantum state, 
\begin{align}
    \sum_x \psi_{\theta}(x) \ket{x} \to \sum_{\omega} \hat{\psi}_{\theta}(\omega) \ket{\omega}.
\end{align}
While the QFT is usually associated with the discrete Fourier transform of the amplitudes as a function on $\mathbb{Z}_N$, this process is known to be efficient in the number of amplitudes for many groups\footnote{The discrete Fourier transform on the amplitudes is performed by the ``standard'' Fourier transform, while a layer of Hadamard gates implements the Walsh or ``boolean'' Fourier transform, and a rather complicated circuit is known to exist for a Fourier transform over the symmetric group.} which implies an exponential---or sometimes even super-exponential---speedup compared to the classical Fast Fourier Transform. Once in the Fourier basis, we can use the entire arsenal of quantum algorithms to manipulate the Fourier spectrum, for example to impose (possibly learnable) biases towards lower-order coefficients. 

Another example of how quantum machine learning relates to spectral methods is given by a prominent class of supervised quantum machine learning models known as ``Quantum Neural Networks'' \cite{farhi2018classification, schuld2020circuit}. These models encode an input vector $x\in \mathbb{R}^N$ into a trainable state $\ket{\psi_{\theta}(x)}$, and define a model as an observable expectation $f(x) = \bra{\psi_{\theta}(x)} O \ket{\psi_{\theta}(x)} $. The standard data-embedding strategy encodes the entries $x_i$ of $x$ via gates of the form $\exp{(i x_i H)}$, where $H$ is a Hermitian operator. The resemblance to Fourier basis functions leads to feature maps that have elegant interpretations via Fourier analysis, and induce different kinds of spectral biases into Quantum Neural Networks \cite{schuld2021effect, lu2026unified, duffy2025spectral, xu2025spectral}, as well as to dequantisation methods \cite{sweke2025kernel, sweke2025potential}. Here, an \textit{embedding strategy}, rather than a QFT, designs the Fourier spectrum of a quantum model.

While not a main focus of this paper, we want to remark that we can also use spectral methods more indirectly, namely to access more sophisticated algebraic properties of the model which relate to its simplicity. A possible example is a recent variation of the famous quantum algorithms solving Hidden Subgroup Problems \cite{Childs_2010} such as Shor's algorithm. Given several copies of a quantum state---which may represent a generative model---we can probe the (un)entanglement structure of the state \cite{bouland2024state, simidzija2026approximate}, which opens up methods of influencing the simplicity of a model in terms of its entanglement. 

We believe that these examples are just the tip of a large iceberg of spectral learning methods unlocked by quantum computers that wait to be uncovered. Our hope is to stimulate more work in this field, and to find satisfying answers to why quantum theory should fundamentally help learning from data---be it classical or quantum. Important questions are, for example: 
\begin{itemize}
    \item Can quantum models design a soft spectral bias without reverting to vast scales used in deep learning?
    \item Can they extend the success of geometric deep learning to implement useful biases, but on intractably large model spaces (rather than on a tractable image)?
    \item Can group structure help to achieve tailor-made regularisation for data domains that standard methods still struggle with, such as graphs, spheres and permutations?
    \item  What limitations arise from the fact that we can only apply Fourier transforms to amplitudes, but not to probabilities directly? What can we do with the sampling access we have to quantum information?
    \item  What are the limitations of classical kernel methods, and can we go beyond them with spectral methods in quantum machine learning?
\end{itemize}  

While these questions are wide open, the goal of this article (which we understand as a mix of perspective, tutorial and review) is to collect the technical tools necessary to start research in the spectral approach to quantum machine learning. 

We will start with a motivating example in Section \ref{sec:toy_example} that intends to give a concrete intuition for the potential and pitfalls of spectral methods in quantum machine learning. Section~\ref{sec:spectral_methods} then lays the mathematical foundations of group Fourier transforms. Section~\ref{sec:critical_ml} gives an overview of spectral methods in classical machine learning, and, for three examples, shows why the Fourier spectrum of a model relates to its simplicity. We define the important tool of quantum Fourier transforms in Section~\ref{sec:natural_quantum} and look a little deeper into Fourier spectra of states that encode quantum models.  Finally, Section~\ref{sec:spectral_qml} gives an overview of existing lines of research that make use of Fourier methods in quantum machine learning. We conclude by commenting on the potential of this approach.

\section{A motivating example} \label{sec:toy_example}

To build intuition, we want to start with a simple toy example of how a generative quantum model can be designed in Fourier space using the QFT. We will work with binary-valued data and a simple non-parametric model---a model that is directly constructed from the training data with no tunable parameters---to generate samples that are similar to the training examples. The idea is to start with the empirical distribution that only has support in the training data, and apply a low-pass filter in Fourier space to ``smoothen'' this distribution and hence remove finite-sample effects. This principle can be implemented with quantum and classical methods, and allows us to illustrate themes that are encountered repeatedly when working with spectral methods in quantum machine learning:
\begin{itemize}
    \item \textit{Desirable model properties are often explicit in Fourier space.} Here we will use the fact that smooth models have a decaying Fourier spectrum, which was already discussed in the introduction.
    \item \textit{We need to capture the data structure by using an appropriate group for the Fourier transform.} In this example we use the group $\mathbb{Z}_2 \times \mathbb{Z}_2 \times \dots = \mathbb{Z}_2^n$, the set $\{0,1\}^n$ with addition modulo $2$. This captures the properties of binary random variables, and is closely related with quantum computation.
    \item \textit{Quantum algorithms can sometimes act directly in Fourier space.} Quantum machine learning models can often move into the Fourier basis, and hence manipulate the Fourier spectrum of the amplitudes of a quantum model using tools from the quantum algorithms toolbox. These tools also define the fundamental limitations on what we can do.
    \item \textit{Quantum methods manipulate the Fourier coefficients of amplitudes, not those of the measured information.} This property, an implication of the Born rule, can be a bug or feature of quantum model design in Fourier space. Here we will see how this property amplifies the model's desired decay in Fourier space, but mismatches lower-order Fourier coefficients.
    \item \textit{The computational cost of quantum and classical algorithms based on Fourier transforms are subtle.} We will show that, while in general classically intractable, a naive spectral design approach can be classically implemented by a kernel method. 
\end{itemize}

%In Section~\ref{sec:spectral_qml} we will see how a variational approach turns the simple model into a potentially quite powerful generative modeling approach \cite{joseph} [TODO: add reference to Joseph].

\subsection{Smoothness for models over binary data}

Let us start by defining a generative learning problem:
\begin{problem}\textbf{Learning to generate bitstrings.}\label{prob:bitstring_learning}
Given a training set of bitstrings $x \in \{0,1\}^n$ sampled from an unknown distribution $p(x)$, generate new bitstrings from that distribution. 
\end{problem}

This learning problem is unsolvable without further assumptions, as lots of distributions have a high probability of generating the training set---how can we guess which was the right one? As discussed in the introduction, an implicit assumption of many machine learning models for real-life problems is that the ground truth is \textit{smooth}. Informally, smoothness means that data which is ``close'' to high-probability samples also have a high probability. Of course, closeness is not always a clear term when dealing with discrete data. We will therefore motivate the \textit{expected parities of subsets of bits} as a measure of smoothness for binary data, which will turn out to be Fourier coefficients of the model $p(x)$, at least if we use the appropriate Fourier transform for the group $\mathbb{Z}_2^n$, the \textit{Walsh-Hadamard transform}.

Intuitively, a probability distribution $p: \{0, 1\}^n \to [0,1]$, is considered smooth if its value does not change drastically when a small number of bits are flipped. We can link this to a property of the \textit{parity functions} 
\begin{align}\label{eq:parity_func}
    \chi_{k}(x) = (-1)^{x\cdot k},
\end{align}
where $k \in \{0, 1\}^n$ is another bitstring, and  $x \cdot k$ is the dot product modulo $2$, e.g., $101 \cdot 111 = (1 \cdot 1 + 0 \cdot 1 + 1 \cdot 1) \mod 2 = 0$. In words, the 1-entries of a given $k$ act to select bits in $x$, and the parity function returns $1$ if there are an even number of 1s among these bits, and $-1$ if there are an odd number of bits. A $k$ is large (and the corresponding parity function is of \textit{high order}) if it contains many $1$s, and hence checks the parity of a large number of bits in $x$. The number of $1$s is called the \textit{Hamming weight} of $k$. For a given $k$ we can define the (normalised) \textit{expected parity} of a distribution as
\begin{align}\label{eq:ft_z2n}
    \mathbb{E}_{p}\left[ \chi_{k}(x) \right] = \frac{1}{\sqrt{2^n}} \sum_{x \in \{0,1\}^n} p(x) (-1)^{x\cdot k}.
\end{align}
The expected parities allow us to make the idea of smoothness of probabilistic models over binary data precise: Smooth bitstring distributions have expected parities that decay for higher orders. This property is easy to intuitively understand. Take, for example, the highest-order $k = 1...1$: the parity $\chi_{1..1}(x)$ is extremely sensitive to noise in $x$, as flipping only one bit randomly will always change the parity. A distribution that has a large value $\mathbb{E}_{p}\left[ \chi_{1..1}(x) \right]$ needs to have very different probabilities for $x$ that are just one bitflip away---it is not very robust.

As mentioned, it turns out that the expected parities $\mathbb{E}_{p}\left[ \chi_{k}(x) \right]$ are the \textit{Fourier coefficients} of $p(x)$ when using the natural Fourier transform over the boolean cube (known as the \textit{Walsh(-Hadamard) transform}). As a consequence, we can assume from now on that smooth bitstring distributions have a decaying Fourier spectrum.

To generalise from training data, a good strategy is to work with a model that has a bias towards smoothness. This is what we will construct next.

\subsection{Smoothing the Fourier spectrum of the data}\label{sec:generative_bandlimiting_classical}

Most generative models in machine learning, including most diffusion models, energy-based models, flows, and variational autoencoders, define a model class, take a random initial model, and train it to maximize the likelihood of seeing the data. We could follow this strategy, and construct a class of functions that has a decaying Fourier spectrum (in fact, this is what a support vector machine does for supervised learning \cite{steinwart2008support}). But we could also do something more direct (see Figure \ref{fig:smoothing_principle})---something that could be elegantly implemented in a quantum model: impose a bias directly to the Fourier coefficients 
\begin{align}\label{eq:emp_fourier}
    \mathbb{E}_{p_{\mathcal{X}}}\left[ \chi_{k}(x) \right] = \frac{1}{\sqrt{2^n}|\mathcal{X}|} \sum_{x \in \mathcal{X}} (-1)^{x\cdot k}
\end{align} 
of the empirical distribution $p_{\mathcal{X}}$ that only has support on the training data $\mathcal{X}\subseteq \{0,1\}^n$. These ``empirical'' Fourier coefficients can be understood as a noisy version of the true ones. Since the empirical distribution is sparse, its Fourier spectrum is dense, and therefore has large higher-order coefficients. We know these higher-order coefficients should be zero for a smooth distribution, and so we can therefore consider them to be artifacts of a finite data set. In particular, since $(-1)^{k\cdot x}$ can be seen as a Bernoulli random variable when $x$ is sampled from $p(x)$, the coefficients \eqref{eq:emp_fourier} are unbiased estimates of the true coefficients  \eqref{eq:ft_z2n} with an absolute error that scales as $1/\sqrt{2^n\vert \mathcal{X} \vert}$. Since higher-order coefficients are effectively zero for the true distribution, this means that the relative error of the empirical Fourier coefficients is much larger at higher order. 

The spectral bias should therefore suppress higher-order coefficients, to ensure a smooth distribution and correct for these finite-data artifacts. On the other hand, since the lower order coefficients are not zero and we do not know their values a priori, it should not change them much beyond the $1/\sqrt{2^n\vert \mathcal{X} \vert}$ estimation error. For example, if all training points start with a 1 bit, it is reasonable to assume that the expected parity of the ground truth for $k=10000$ is close to 1, which means that new samples also have a 1 bit in this position. The idea is hence to ``denoise'' the data by applying something that would be called a \textit{low-pass filter} in signal processing. Note that while we will do this with a ``hard'' strategy in this didactic example, ultimately such a bias will have to be learnt to build performant models (which is why we suggestively denote the final bandlimited model by a subscript $\theta$).

\begin{figure}
    \centering
    \includegraphics[width=\linewidth]{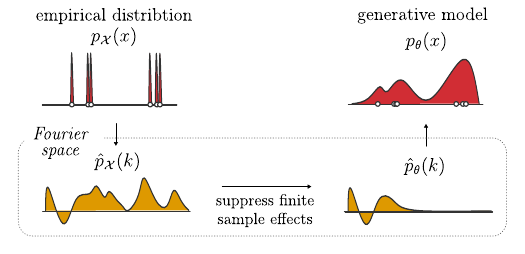}
    \caption{A sketch of the idea of smoothing an empirical distribution of training samples in Fourier space. The empirical distribution is sparse, with support only on the training data. Therefore, its Fourier spectrum is dense and has support on high-order frequencies which can be seen as a consequence of finite data effects. By applying a low-pass filter in Fourier space we impose smoothness on the resulting distribution.}
    \label{fig:smoothing_principle}
\end{figure}

The considerations so far suggest the following spectral strategy to solve the generative learning task:

\begin{model}{\textbf{Empirical smoothing}}
We are given a set $\mathcal{X}$ of data samples and want to solve Problem \ref{prob:bitstring_learning}. 
\begin{enumerate}
\item Start with the empirical distribution
\begin{align}
p_{\mathcal{X}}(x) =  \frac{1}{|\mathcal{X}|} \sum_{x \in \{0,1\}^n} \mathbb{I}[x \in \mathcal{X}],
\end{align}
where $\mathbb{I}$ is the indicator function that is one if $x$ is a training sample, and zero elsewhere.
\item Move into the $\mathbb{Z}_2^n$ Fourier space using the Walsh transform
\begin{equation}
\hat{p}_{\mathcal{X}}(k) = \frac{1}{\sqrt{2^n}}\sum_{x \in \mathcal{X}} \frac{1}{|\mathcal{X}|} (-1)^{k \cdot x}.
\end{equation}
\item Apply a (possibly learnable) filter in Fourier space that suppresses higher-order Fourier coefficients, but keeps the lower-order Fourier coefficients approximately intact. 
\item Move back to direct space.
\item Sample from the model.
\end{enumerate}
\end{model}
For example, we could apply a decay that depends on the Hamming weight of the Fourier coefficient, which is implemented in Figure~\ref{fig:smoothing_classical}.

\begin{figure*}
    \centering
    \includegraphics[width=0.9\linewidth]{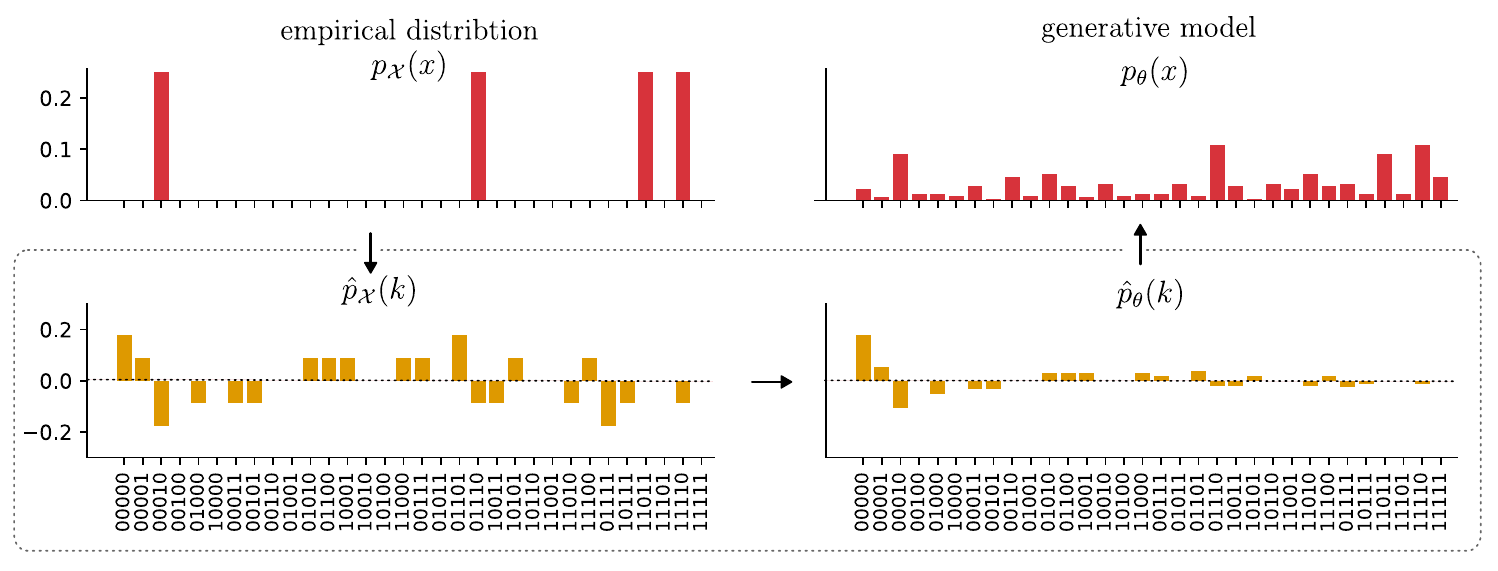}
    \caption{An example of smoothing an empirical distribution in Fourier space to solve the generative learning problem. Starting from an empirical distribution $p_{\mathcal{X}}$ of four binary data samples (top left), we compute the (Walsh) Fourier spectrum (bottom left). We then apply a filter of the form $\hat{p}_{\theta}(k) = (1-2\theta)^{|k|} \hat{p}_{\mathcal{X}}(k)$, where $\theta \in [0,1]$ controls the spectral decay rate. Moving back into direct space, we get a model that generalised from the training data by making the empirical distribution smoother. }
    \label{fig:smoothing_classical}
\end{figure*}

When we move back into direct space, we get a different (and typically dense) probability distribution---we have generalised from the training data (see Figure~\ref{fig:smoothing_classical}). In fact, since for boolean functions the Fourier decomposition is a kind of polynomial decomposition, what we just did is a (hard, and not learnable) order-based regularisation strategy as suggested in \cite{wilson2025deep}. 

% Quite clearly, the power of the model lies in the ability to apply good filters in Step 3. According to Bochner's theorem, such a transformation has to fulfill certain mathematical properties to make sure that the final model is a probability distribution: 

% \begin{theorem}{\textbf{Properties of the filter}}
% Let $p(x)$ be a probability distribution over bitstrings $x \in \{0,1\}^n$, and $\hat{p}(k)$ its (Walsh) Fourier coefficients. We transform each coefficient by applying a filter $\to \hat{p}'(k) = \hat{g}(k)  \hat{p}(k)$. The resulting $\hat{p}'$ form the Fourier spectrum of a probability distribution iff:
% \begin{enumerate}
%     \item \emph{Normalisation:} $\hat{g}$ does not change the zero frequency, or $\hat{p}(0..0) = \hat{p}'(0..0) =1$.
%     \item \emph{Positivity:} The new spectrum $\hat{p}'(k)$ is positive definite.
% \end{enumerate}
% If $\hat{g}(k)$ only depends on $k$, the second requirement means that it has to be positive definite.
% \end{theorem}

\subsection{Implementation with quantum models}

The strategy sketched above can be easily translated into a quantum algorithm:
\begin{model}{\textbf{Empirical smoothing with quantum models}}
We are again given a set $\mathcal{X}$ of data samples and want to solve Problem \ref{prob:bitstring_learning}, but this time we use a quantum state as a generative model. 
\begin{enumerate}
    \item Start with a superposition of the training data, 
    \begin{equation}
        \ket{\psi_{\mathcal{X}}} = \frac{1}{\sqrt{|\mathcal{X}|}} \sum_{x \in \mathcal{X}} \ket{x}
    \end{equation}
    \item Apply the $\mathbb{Z}_2^n$ Fourier transform to this state, which is implemented by applying a Hadamard gate onto each qubit,
    \begin{equation}
        \ket{\hat{\psi}_{\mathcal{X}}} = \frac{1}{\sqrt{2^n |\mathcal{X}|}} \sum_{k} \sum_{x \in \mathcal{X}}  (-1)^{x\cdot k} \ket{k}.
    \end{equation}
    \item Apply a (potentially non-unitary) quantum transformation that modifies the amplitudes of the quantum state in Fourier space. 
    % \item Apply a quantum algorithm (unitary or non-unitary operation) to shift weight from amplitudes of high-order coefficients $\ket{k}$ to lower-order ones, while keeping the relative distribution of the lower-order ones as intact as possible. 
    \item Move back into direct space by applying an inverse Fourier transform, another layer of Hadamards.
    \item Measure in the computational basis to generate a sample from the model.
\end{enumerate}
\end{model}

It is important to note that manipulating the Fourier spectrum of the \textit{amplitudes} $\psi(x)$ (Step 3) of the quantum state is an indirect way to manipulate those of the \textit{measurement distribution} (i.e., the final probabilistic model), which is given by the Born rule, 
\begin{align}
    p(x) = |\psi(x)|^2.
\end{align}
% In fact, the convolution theorem implies that the two are related via auto-correlation:
% %
% \begin{align}
%     \hat{p}(k) = \frac{1}{\sqrt{2^n}}\sum_x \hat{\psi}(x)\hat{\psi}^*(x \oplus k)
% \end{align}
The Fourier coefficients of $\psi(x)$ and $p(x)$ are related by an autocorrelation, as we will see in Section \ref{sec:qft_generative_qml}. Comparing Figures~\ref{fig:smoothing_classical} and \ref{fig:smoothing_amplitudes} shows that the spectral manipulation of amplitudes leads to an attenuated profile of the final model compared to the same manipulation being applied to probabilities. Furthermore, unitary transforms on the amplitudes can lead to \textit{non-linear} transformations of the distribution's Fourier spectrum. Depending on the application, working with amplitudes can therefore be a bug or feature for model design.

Of course, quantum computing imposes fundamental constraints on how to manipulate the Fourier spectrum of the state in Step 3. For example, we can naively implement the noise filter used in Figure~\ref{fig:smoothing_classical} on the amplitudes as shown in Figure~\ref{fig:smoothing_amplitudes} by using a quantum algorithm that conditionally rotates an ancilla state by an angle of $2 \arcsin(1-2\theta)$ when a $1$ is encountered in  $\ket{k}$. Post-selecting on the ancilla, we end up with a state that re-weighted the amplitudes in Fourier space according to the noise filter. A fundamental limitation of this simple approach is the success probability of post-selection, which prohibits arbitrarily strong bandlimiting at large scales. More sophisticated techniques \cite{guo2024nonlinear, rattew2023non} may be able to push the boundaries, but always incur a price for highly non-unitary operations. 

A more successful approach may involve opting for a unitary layer that---rather than attempting to bandlimit the amplitudes of the quantum state probabilistically---applies a deterministic phase mask to the amplitudes so that the resulting interference from  autocorrelation results in an effective low-pass filter on the spectrum of the model's probability distribution. 

Ultimately, a flexible model also requires the transformation to be learnt, which suggests the use of scalable variational circuits that come with an additional set of challenges.\footnote{In upcoming work we study a learnable diagonal phase mask implemented by trainable IQP circuits \cite{recio2025train} as a suitable ansatz to build filters.} In summary, quantum computers can efficiently manipulate the Fourier spectrum of the amplitudes of a quantum model state, but if we can design good Fourier filters within the constraints of quantum algorithms will be a crucial question for spectral methods in quantum machine learning.

\begin{figure*}
    \centering
    \includegraphics[width=0.8\linewidth]{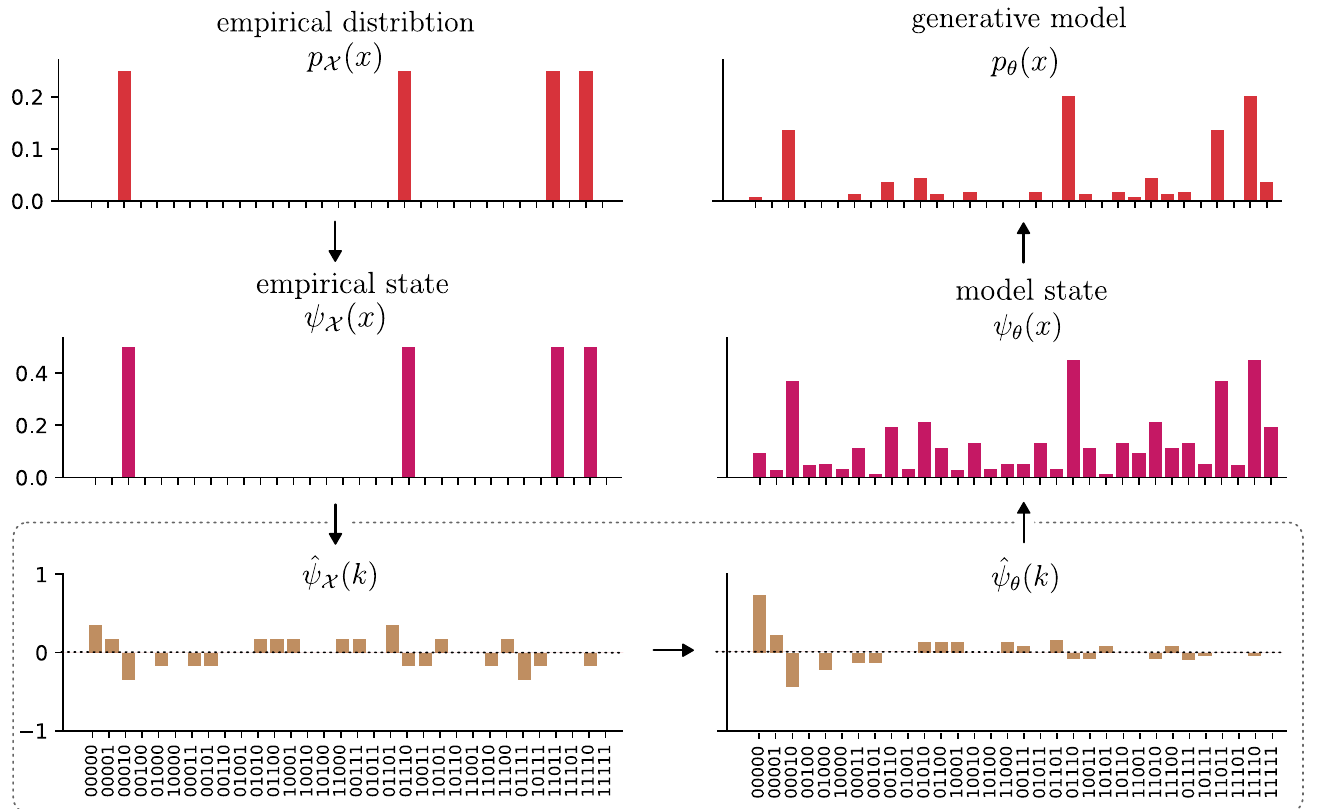}
    \caption{An example of smoothing empirical distributions with quantum computers, using the same filter as in Figure~\ref{fig:smoothing_classical}. Fourier filtering is performed on the amplitudes of the quantum state. We encode the empirical distribution (top left) into a superposition $\ket{\psi_{\mathcal{X}}}$ (mid left), to which we apply the (Walsh) Fourier transform (bottom left). We then apply a non-unitary transformation (with the help of ancillas) to implement the filter 
    $\hat{\psi}_{\theta}(k) \propto (1-2\theta)^{|k|} \hat{\psi}_{\mathcal{X}}(k)$ (bottom right).  Moving back into direct space (mid right), we get a quantum model with a new distribution of amplitudes, whose measurement distribution constitutes the generative model. 
    Comparing to Figure~\ref{fig:smoothing_classical}, it is clear that manipulating the Fourier coefficients of the amplitudes of a quantum model leads to a different generative model, but with a similar structure: high probabilities are amplified and low ones suppressed.}
    \label{fig:smoothing_amplitudes}
\end{figure*}

\subsection{Why empirical smoothing may be classically hard}

We saw that our toy example for a spectral method, generative empirical smoothing, can at least in principle be implemented efficiently on a quantum computer. But how computationally feasible is this strategy classically? In these last two sections we will argue that brute force approaches are doomed to fail, but that care has to be taken when claiming speedups, as some situations allow for implicit ways to efficiently implement empirical smoothing. 

First, let us have a look at efficiently computing probabilities $p(x)$ after classical empirical smoothing. This ability would unlock training a classical model, for example by taking the model class of bandlimited empirical distributions and learning the bandlimiting filter using maximum likelihood estimation. Computing the $k$'th Fourier coefficient of the empirical distribution, 
\begin{align}\hat{p}_{\mathcal{X}}(k) = \frac{1}{\sqrt{2^n}} \sum_{x \in \mathcal{X}} (-1)^{k \cdot x},\end{align} 
is efficient, as it is a sum over tractably many terms. Of course, there are exponentially many such Fourier coefficients. Computing a single probability from an empirical distribution that has been bandlimited to contain only a small number of terms,
    \begin{align}
     p_{\theta}(x) = \frac{1}{\sqrt{2^n}} \sum_{k \in \Omega} \hat{g}(k)\hat{p}_{\mathcal{X}}(k) (-1)^{k \cdot x},
    \end{align}
is tractable. Likewise, computing marginal distributions becomes tractable (by sampling over subsets of Fourier coefficients), which allows us to compute conditional probabilities, opening up the ability to sample from $p_{\theta}(x)$ in an autoregressive manner. Is empirical smoothing then classically tractable?

The issue is the \textit{number} of Fourier coefficients to keep track of. For high dimensions, bandlimiting would have to be limited to a few frequencies in each dimension. For example, in the slowest-growing case of $\mathbb{Z}_2^n$, where there are two frequencies in each dimension, the number of Fourier frequencies with Hamming weight $|k|$ less than some threshold $b$ is given by $\sum_{m=0}^b \binom{n}{m}$. For $n=10,000$ dimensions and $b=2$, we would already have to sum over $50$ million Fourier coefficients to compute the probability of a data point, even for the simplest bandlimited model. For other groups and larger thresholds, this quickly becomes unfeasible. Furthermore, to keep empirical smoothing classically tractable, the filter for the $k$'th coefficient cannot depend on the Fourier coefficients other than $\hat{p}(k)$, as there are intractably many of these (even for the empirical distribution). These two problems clearly do not occur for the quantum computer, which means that empirical smoothing might display quantum speedups for the useful task of generalisation. Of course, whether this strategy is able to solve problems that other approaches to generalisation cannot is an open question.

As a cautionary tale, we want to conclude our analysis of this toy example by showing that for a specific filter it \textit{can} be possible to efficiently sample from a bandlimited empirical distribution on a classical computer. The trick is to use convolution with a kernel, a method that we will discuss later as the standard classical approach for spectral modeling, and hence the most viable ``classical competitor'' for this approach to QML.

\subsection{Example of a dequantisable filter}\label{sec:generative_bandlimiting_kernel}

While we argued that the direct manipulation of the Fourier spectrum required in~\ref{sec:generative_bandlimiting_classical} is not computationally feasible in general, it might very well be possible for very specific Fourier filters. To illustrate this, consider the noise filter used as an example in Figure~\ref{fig:smoothing_classical}:
\begin{align}
\hat{p}_{\theta}(k) = (1-2\theta)^{|k|} \hat{p}_{\mathcal{X}}(k),  
\end{align}
where the hyperparameter $\theta \in [0,1]$ controls the decay rate. Going back into direct space, we get the model
\begin{align}
    p(x) = \sum_{x \in \{0,1\}^n}  p_{\mathcal{X}}(x) \kappa(x-y),   
\end{align}
with $\kappa(x, y) = \theta^{d_H(x, y)}  (1-\theta)^{n - d_H(x, y)}$, which is known as the \textit{noise kernel}. The model can therefore be implemented by a convolution $p_{\mathcal{X}} * \kappa$ in direct space. 
At the same time, $p(x)$ is a typical \textit{kernel density estimation} model, a class of models which define probability distributions as sums of kernel functions centered around the training data. Crucially (and this is only true in rare cases), the kernel has the property that we can sample from $p(x)$ even in high dimensions: we simply sample a training data point and flip each bit with probability $\theta$. This example shows that for some manipulations in Fourier space, kernel methods can even be used for generative models. Clearly, there are a lot of restrictions on the kernel to make this work, but kernel methods have already been used to dequantise a range of supervised quantum machine learning models \cite{sweke2025kernel, sweke2025potential} and are a contender for spectral methods.  

In summary, the simple example of empirical smoothing showed how (group) spectral methods in machine learning can be interesting for machine learning and natural for quantum computers. But it also highlighted subtleties that have to be handled when designing quantum machine learning models---such as the Born rule, the difficulty of manipulating amplitudes, and possible classical dequantisation with kernel methods. These challenges appear in many (but not all) spectral approaches to quantum machine learning. With this illustrative case in mind, let us now get back to the main argument: why quantum computers might be useful to implement spectral methods in machine learning.

\section{Spectral methods} \label{sec:spectral_methods}

As mentioned, we loosely consider \textit{spectral methods for machine learning} as those that design  machine learning models by manipulating their properties in Fourier space. In this section we will rigorously define what we mean by terms like  ``Fourier spectrum'' and ``Fourier coefficients'', in particular in their group-theoretic generalisation. 

\subsection{Motivation}

What do Fourier transforms have to do with groups? Everything, if one looks closely. Let us motivate this with one of the most well-known Fourier transforms, the \textit{discrete Fourier transform}. It transforms a sequence $f_1,...,f_d$ of complex numbers into another sequence of complex numbers. Sometimes the complex values are written as a function $f(x)$ evaluated or ``sampled'' at regular intervals, for example using integer values $x \in \{0,...,d-1\}$. The Fourier coefficients are then given as
\begin{align} 
\hat{f}(k) = \frac{1}{\sqrt{d}}\sum_{x=0}^{d-1} f(x) e^{2 \pi i  \frac{k x}{d}},
\end{align}
where the frequencies $k$ are also in $\{0,...,d-1\}$. 
The inverse Fourier transform is given as 
\begin{align} 
f(x) = \frac{1}{\sqrt{d}}\sum_{k=1}^{d-1} \hat{f}(k) e^{-2 \pi i  \frac{k x}{d}},
\end{align}

Note that the normalisation factor is a matter of convention, and we will here always use the ``balanced'' version which uses the same pre-factor for the Fourier transform and its inverse, and keeps the $L_2$ normalisation of the Fourier coefficient constant. The expressions $e^{2 \pi i  \frac{k x}{d}}$ correspond to Fourier basis functions with integer-valued frequencies $k$, and a Fourier coefficient $\hat{f}(k)$ can be seen as the projection of $f(x)$ onto the $k$'th basis function. The function beyond the interval $0,..,N-1$ is thought to be ``periodically continued'', which means that $f(x) = f(x + N)$. 

But why do---at least if we do not want to incur additional headaches---the $x$ values have to be equidistant? Why this notion of ``periodic continuation''? Why are the basis functions of exponential form? Why are the $k$ also integers? And what makes the Fourier basis special? It turns out that all of these questions have an elegant answer if we interpret the function domain as a group.
As a reminder, a group is a set of elements that has:
\begin{enumerate}
    \item an operation that maps two elements $a$ and $b$ of the set into a third element of the set, for example $c = a + b$,
    \item an ``identity element'' $e$ such that $e + a = a$ for any element $a$
    \item an inverse $-a$ for every element $a$, such that $a + (-a) = e$.
\end{enumerate}
An \textit{Abelian group} has the property that $gg' = g'g$. Its \textit{characters} are functions $\chi:G \to \mathbb{C}$ with the property $\chi(g g')= \chi(g)\chi(g')$.

Returning to the discrete Fourier transform, we can consider the $x$ as elements from the set of integers $\{0,...,d-1\}$, together with a prescription of how to combine two integers to a third from this set. Choosing ``addition modulo N'' for
this operation (which means that $(d-1)+1 = 0$), we get the \textit{cyclic group} $\mathbb{Z}_d$. This choice explains the equidistant property: integers are by nature equally spaced in $\mathbb{R}$. It also explains the periodic continuation, as $x = x + d$ implies $f(x) = f(x +d)$. Furthermore, the integer-valued frequencies $k$ turn out to be elements from the so-called ``dual group'' $\hat{G}$, which in this case looks exactly like the original one. Finally, the Fourier basis functions are exactly the characters of $\mathbb{Z}_d$ (which form a basis of the space of functions on $G$), and the Fourier coefficients hence \textit{the projection of $f(x)$ onto the characters}. 

Note that other standard Fourier transforms can be associated with groups as well: A \textit{Fourier series} transforms periodic functions on the real line with period $2\pi$, which corresponds to the domain $\mathbb{R}/\mathbb{Z}$, while a \textit{continuous Fourier transform} on real numbers relates to the group $(\mathbb{R}, +)$ (the real numbers under addition). Multivariate functions, like those over $\mathbb{R}^N$, are then defined over direct products of groups.

\subsection{Group Fourier transforms}

For simplicity, we will focus on discrete groups---which are most relevant for qubit-based quantum computing---and follow the insightful presentation in \cite{kondor2008group}. Continuous Fourier transforms over locally compact groups simply turn the sum into an integral over a Haar measure on $G$.

For discrete Abelian groups, the standard Fourier transforms generalise as follows: 
\begin{definition}{\textbf{Fourier transform over Abelian groups}}
Let $G$ be a discrete Abelian group and $\chi_k$ its characters from the Pontryagin Dual group $\hat{G}$. The Fourier transform $\mathcal{F}$ maps a function $f:G \to \mathbb{C}$ to a function $\hat{f}:\hat{G} \to \mathcal{C}$:
    \begin{align}
        \hat{f}(k) = \frac{1}{\sqrt{|G|}} \sum_{g \in G}  f(g) \chi^*_k(g).   
    \end{align}
The inverse Fourier transform, implementing the reverse map, computes the function
\begin{align}
f(g) = \frac{1}{\sqrt{|G|}}\sum_{k \in \hat{G}} \chi_k(g) \hat{f}(k).
\end{align}
\end{definition}
Again, the normalisation of forward and reverse Fourier transform, and which one carries the conjugate of the character, is a convention, and sometimes $ \chi_k(g^{-1}) = \chi^*_k(g)$ is used. 

The Fourier transform over a discrete, locally compact \textit{non}-Abelian group looks slightly different from the Abelian Fourier transform. Instead of the characters of the group, they invoke objects called the \textit{irreducible representations} or \textit{irreps}. A \textit{representation} is a map $R:G \rightarrow GL(V)$ from group elements to linear transforms on a vector space $V$. There are many representations for a given group, each referring to a different vector space. A subspace $W \subseteq V$ is called \textit{invariant} if for every group element $g \in G$ and every vector $w \in W$, the result $R(g)w$ is still inside $W$. The representation is said to be \textit{irreducible} if the only invariant subspaces are the trivial zero space $\{0\}$ and the entire space $V$ itself.

With this, we can define the most general version of the group Fourier transform as follows:
\begin{definition}{\textbf{Fourier transform over non-Abelian groups}}\label{def:FT_nonabelian}
Let $G$ be a discrete non-Abelian group and $\mathcal{R}$ the set of its (matrix-valued) inequivalent irreducible representations. The Fourier transform maps a function $f:G \to \mathbb{C}$ to the matrix-valued Fourier coefficients
    \begin{align}
        \hat{f}(\sigma) = \frac{1}{\sqrt{|G|}} \sum_{g \in G}  f(g) \sigma(g)^{\dagger},   
    \end{align}
with $\sigma \in \mathcal{R}$.
The inverse Fourier transform is given by
\begin{align}
f(g) = \frac{1}{\sqrt{|G|}} \sum_{\sigma \in \mathcal{R}} d_{\sigma} \mathrm{tr}\{ \hat{f}(\sigma) \sigma(g)\},
\end{align}
with $d_{\sigma}$ the dimension of the irrep $\sigma$.
\end{definition}
In the non-Abelian case, the Fourier coefficients are matrices of varying sizes (and hence basis dependent, which can complicate their interpretation). Note that the Abelian case is a special case of Definition~\ref{def:FT_nonabelian}, as the irreps of an Abelian group are one-dimensional, and hence equal to their traces, i.e., the characters.

\subsection{Convolution}

Given the prominent role of the convolution theorem 
\begin{align}
    \widehat{f \star g }(\sigma) = \hat{f}(\sigma) \hat{g}(\sigma),
\end{align}
it is useful to also define the group-theoretic generalisation of convolution:
\begin{align}\label{eq:convolution_group}
    (f \star g)(g) = \sum_{g' \in G} f(g (g')^{-1}) g(g').
\end{align}
For Abelian groups, where the group operation is usually denoted by a $+$, while adding an inverse element becomes a $-$, this becomes the more familiar 
\begin{align}\label{eq:convolution_standard}
    (f * g)(x) = \sum_{y \in \mathbb{R}} f(x - y) g(y).
\end{align}
Note that in machine learning, convolution is often conflated with the more common \textit{cross-correlation}, which replaces the minus sign with a plus.

\subsection{Symmetry of the Fourier basis}

The inverse Fourier transform can be understood as a decomposition of a function over $G$. According to the Peter-Weyl theorem, the irrep entries $\{\sigma_{ij}\}_{\sigma \in \mathcal{R}, ij \in \{1,\dots, d_{\sigma}\}}$ form a complete orthogonal basis of $L^2(G)$, the space of square integrable functions on the group. But what is so special about the Fourier basis? A defining feature of the Fourier basis is that the subspaces $L_\sigma \subseteq L^2(G)$ spanned by the set of all irrep entries $\{\sigma_{ij}(g)\}$ for a given $\sigma$ are \textit{invariant under the group action}. This follows from the definition of a representation:
\begin{align}
\sigma(g g') = \sigma(g) \sigma(g'),
\end{align}
which means that a translation by $g'$ does not cause the function to leave the invariant subspace 
\begin{align}
    L_{\sigma} = \text{span}\{\sigma_{11}(g), \sigma_{12}(g), \dots,\sigma_{d_{\sigma}, d_{\sigma}}(g)\}.
\end{align}
For Abelian groups (using additive notation below), where we index the dual group by $k\in\hat{G}$, the irreps are one-dimensional. This means the invariant subspaces $L_k = \rm{span}\{\chi_k(g)\}$ are each spanned by a single character that fulfills 
\begin{align}
\chi_k(g + g') = \chi_k(g) \chi_k (g'),
\end{align}

The significance of this invariance property is discussed at length in the work of Diaconis \cite{diaconis1988group}. For example, if the symmetric group acts by permuting the features of a data vector, this invariance means that the Fourier power spectrum (the energy captured within each invariant subspace) does not change if we permute the order of the features, which is often an arbitrary labeling decision. This is intimately related to the smoothness inherent in the Fourier spectrum: it informs us how the function transforms if we shift it in direct space. In this sense, the Fourier basis can capture important structure in the data under ``shifts'', which helps quantify its simplicity for the design of machine learning models---a claim that we will support with more evidence in the next section.

\subsection{Fourier transforms on homogeneous spaces}

While we focus on spectral methods for functions on groups in this paper, there is an important generalisation to functions on the \textit{homogeneous space} $X$ that a group $G$ \textit{acts on}. (We will denote a group action as $\rhd: G \times X \rightarrow X$ with $g \rhd x = x'$.) This opens up the realm of spectral analysis to many more data types, which is why we will briefly list the relevant definitions here. 

Technically, a homogeneous space $X$ is a space where the group $G$ acts transitively (see Figure~\ref{fig:lifting}). This means one can get from any point $x \in X$ to any other point $y \in X$ using a group operation, $g \rhd x = y$. Let us denote as $H$ the \textit{stabilizer subgroup} of an arbitrary reference point $x_0 \in X$,
\begin{align}
H = \{ h \in G | h \cdot x_0 = x_0 \}.
\end{align}
The cosets $gH = \{gh | h \in H\}$ of the stabiliser subgroup are sets of group elements whose action maps $x_0$ onto the same point $x \in X$. This means that each point in $X$ can be associated with a coset of $H$ in $G$. This allows us to associate $X$ with the quotient space, $X \cong G / H$.

\begin{figure}
    \centering
    \includegraphics[width=0.65\linewidth]{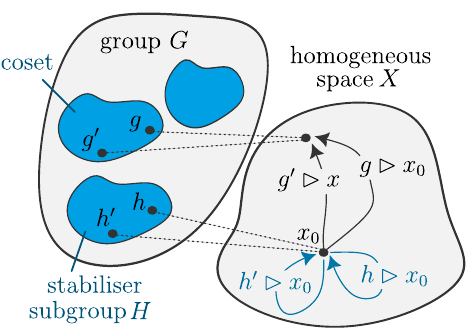}
    \caption{Illustration of the relation between a group and a homogeneous space, which are spaces where every point can be reached by acting with a group element. Group elements that map to the same point live in cosets of the stabiliser subgroup (whose elements map $x_0$ to itself). We can associate the homogeneous space with the quotient space $G/H$.}
    \label{fig:lifting}
\end{figure}

To spectrally analyze a function $f: X \to \mathbb{C}$, we can therefore interpret it as a function $f: G/H \to \mathbb{C}$ and then ``lift'' it to $G$ by making it constant on all group elements in the same coset (i.e, the blue regions in Figure~\ref{fig:lifting}). We define the \textit{lifted function} $f^{\uparrow}: G \to \mathbb{C}$ by 
\begin{align} 
f^{\uparrow}(g) = f(g \rhd x_0).
\end{align}
Mathematically, this symmetry means that $f^{\uparrow}$ is right-invariant under $H$:
\begin{align}
    f^{\uparrow}(gh) = f(gh \rhd x_0) = f(g \rhd x_0) = f^{\uparrow}(g) \quad \forall h \in H.
\end{align}

We can now apply the standard Fourier Transform to $f^{\uparrow}$: 
\begin{align} 
\hat{f^{\uparrow}}(\rho) = \sum_{g \in G} f^{\uparrow}(g) \rho(g)^{\dagger} .
\end{align}
The technique of lifting has enabled the use of group Fourier transforms to rigorously identify convolutional layers with group-equivariant transformations, which gave rise to generalisations of convolutional neural networks \cite{kondor2018generalization} in machine learning research. It is also related to one of the first formulations of hidden subgroup problems, which Kitaev phrased as an Abelian stabiliser problem in 1995 \cite{kitaev1995quantum}. 
In the following we will continue to focus on Fourier transforms of functions on groups rather than homogeneous spaces, but acknowledge the potential that lies in this simple extension of spectral analysis.

\section{Crucial for machine learning}\label{sec:critical_ml}

Smoothness is desirable in machine learning because it leads to robust models that generalise well \cite{pmlr-v137-rosca20a}, and has motivated many methods \cite{JMLR:v15:srivastava14a,noh2017regularizing,miyato2018spectral,dherin2022why} that aim to bias models towards smoothness. Yet, these techniques are scattered and act indirectly on the model spectra, suggesting the need for a broader paradigm for directly imposing simplicity biases. Thus, before discussing quantum machine learning algorithms, we have to address the elephant in the room: Why should Fourier space be a prime place to design the learning properties of a machine learning model? The argument of this paper, which we present in more detail here, is that machine learning fundamentally needs to impose biases towards ``simple models'', and that the Fourier spectrum captures what is meant by ``simple''---for ``conventional'' data types as well as for more specific, group-structured data such as permutations. 

To that end, we will first review the role of simplicity and regularisation in learning theory. We will then take a closer look at the notion of simplicity captured by the low-order part of the Fourier spectrum for three example groups, which showcases the diversity in which ``simple'' can be interpreted. Lastly, we will highlight that spectral methods are already quite prominent in machine learning even if they are hardly ever explicitly mentioned in textbooks: they are implicitly used in methods that involve (stationary) kernels, such as convolutional neural nets, support vector machines, the maximum mean discrepancy to train Generative Adversarial Networks (GANs), and the Neural Tangent Kernel that became a prime framework to investigate the working principle of of deep learning. 

\subsection{Why a simplicity bias is crucial for machine learning }

The foundations of learning theory \cite{vapnik95, valiant1984theory} cast machine learning as the problem of balancing expressivity and simplicity in models built from data. This is made particularly concrete in the framework of \textit{statistical learning theory} developed by Vapnik and Chervonenkis \cite{vapnik1971uniform}, which quantifies the generalisation power of a supervised model through the principle of \textit{empirical risk minimization}. In this framework, the goal of machine learning is to learn some function $f$ from a set of possible functions $\mathcal{H}$ that minimizes the \textit{expected risk} 
\begin{align} 
R(f) = \int_{X}  \; p(x, y) L(f(x), y) \; dxdy, 
\end{align} 
which represents the model's performance on unseen data drawn from a ground truth distribution $p(x, y)$ with respect to some loss metric $L$. Practically, $R(f)$ is exactly what the \textit{test error} is supposed to measure. Since $p$ is unknown, machine learning practitioners try to minimise the expected risk by minimizing the \textit{empirical risk} 
\begin{align} 
\hat{R}_M(f) = \frac{1}{M}\sum_{x \in \text{data}} L(f(x), y), 
\end{align} 
the average error on a training set of size $M$. Importantly for us, learning theory typically upper bounds the expected risk by the empirical risk, plus a term $C(\mathcal{H})$ that accounts for the \textit{capacity} of the model \cite{blumer1989learnability,hastie2009elements},
\begin{align}
    R(f) \leq \hat{R}_M(f) + C(\mathcal{H}).
\end{align}
To get a small expected risk (or test error), one has to decrease both terms on the right hand side. However, these terms are at odds with each other. This is sometimes captured in the ``bias-variance-tradeoff'', a term that refers to the high variance of models that have a low training error (as they overfit the particular data, and produce very different predictions when trained on different training data sets), and a high bias of those that are simple (as they cannot learn the intricacies of the pattern).
Learning, as a conclusion, is the art of balancing the model capacity with its flexibility.

The advent of deep learning has brought some more nuance into this picture. In the past, the capacity term was chosen to capture the expressivity of the model's function class, such as the \textit{VC-dimension} \cite{vapnik1971uniform} that measures how flexible the decision boundary of the best function in the model class is, or the \textit{Rademacher complexity} \cite{bartlett2002rademacher} that measures the average capacity of a model class to classify randomly assigned labels. When deep learning showed increasing evidence of empirical success, one of the biggest ``mysteries'' was why large neural networks---a model class that has an arbitrarily large term $C(f)$ under these measures, as it can learn even unstructured labels \cite{zhang2016understanding} --- generalises so well. The resolution was found in the interplay of large models, big data and the optimisation algorithm (i.e., \cite{belkin2019reconciling, jacot2018neural,  kalimeris2019sgd, keskar2016large}): somehow, when these come together, the models that are effectively found in training are simple, even though training \textit{could} in principle find much more complex models that do not generalise. In other words, deep learning has a hidden, or ``implicit'' regularisation bias. This bias seems to---somewhat paradoxically---fundamentally rely on scale: the \textit{lottery ticket hypothesis} \cite{frankle2018lottery} suggests that while \textit{trained} neural networks are rather compressible (i.e., ``simple''), they cannot be found using smaller architectures that are tailor-made to the solution. Altogether, the insights from deep learning were well summarised by Wilson's \textit{soft simplicity bias}:
\begin{quote}``rather than restricting the hypothesis space to avoid overfitting, [good models] embrace a flexible hypothesis space, with a soft preference for simpler solutions that are consistent with the data''. \cite{wilson2025deep}\end{quote}

\subsection{How a Fourier spectrum defines simplicity}\label{sec:interpreting_spectrum}

If we understand machine learning as the problem to find expressive, scalable model classes that have a simplicity bias, a major component of designing good models is to engineer the simplicity bias. For deep learning, this happened somewhat accidentally, and relies on heavy computational resources. An important open question is hence how to engineer ``soft simplicity biases'' more consciously, and hopefully without the immense energy requirements of large neural networks.

We will now explore in more depth why the Fourier spectrum provides a useful mathematical framework to define---and subsequently exploit---notions of smoothness. This connection is relatively straightforward for the standard Fourier transform on $(\mathbb{R},+)$, where the coefficients capture smoothness, and on $\mathbb{Z}_2^N$, where they capture correlations and interaction effects. Building on this, we will present an intuition for the Fourier spectrum of the non-Abelian symmetric group $\mathbb{S}_n$, which captures significantly more complex correlation patterns. 

Before proceeding, we emphasize that our notion of smoothness in both the Abelian and non-Abelian cases depends on some assumptions. While we may always define smoothness from first principles in terms of a particular set of generators, we generally rely on empirical data and prior beliefs in order to choose a natural set of generators that best describes our observations in terms of smooth functions. In this way, the treatment considered here may be generalised to a large class of group-theoretic settings.

\subsubsection{Continuous features}

\begin{figure}[t]
    \centering
    \includegraphics[width=0.8\linewidth]{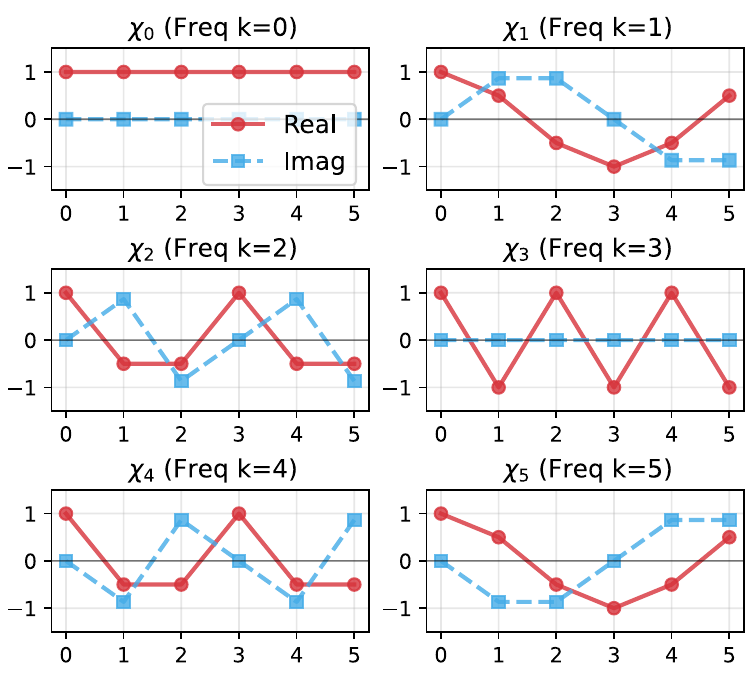}
    \caption{Fourier basis functions $\chi_k(x)$ for the cyclic group $\mathbb{Z}_d$ with imaginary part in blue and real part in red.}
    \label{fig:zd_characters}
\end{figure}

The standard \textit{continuous Fourier transform}, which we mentioned in the introduction, transforms functions on the group $\mathbb{R}$ (or more precisely, $(\mathbb{R}, +)$), with characters $\exp{(2 \pi i x k)}$. Its multidimensional version maps functions on products of this group, $\mathbb{R}^N$. The Fourier coefficients of a function over $\mathbb{R}^N$ are related to its \textit{smoothness}.

The most obvious connection between the Fourier spectrum and smoothness of a bounded function on $\mathbb{R}^N$ is given by the canonical definition of smooth functions as infinitely differentiable. One can show that such functions have a super-polynomial asymptotic decay in Fourier space, or
\begin{align}
    |\hat{f}(\omega)| = \mathcal{O}\left( |\omega|^{-m} \right) \; \text{as} \; |\omega| \to \infty.
\end{align}
Using the formula for partial integration, and assuming that $f$ decays at $\pm \infty$, one can write a Fourier coefficient $\hat{f}(k)$ as the Fourier coefficient of the derivative of $f$,
\begin{align}
    \hat{f}(k) &= \int dx f(x) e^{-2 \pi i kx}\\
    &=  \frac{1}{(2\pi i k)^m}\int dx \left( \partial^m_x f(x) \right) e^{-2 \pi i kx}\\
    &= \frac{1}{(2\pi i k)^m} \widehat{\partial^m_x f(x)}.
\end{align}
Bounding $|\widehat{\partial^m_x f(x)}|$ by the finite $L_1$ norm of the input function here one sees that
\begin{align}
    |\hat{f}(k)| \leq c/k^m
\end{align}
for some constant $c$, which is a spectral decay with the power of $m$. If $m$ goes to infinity, the asymptotic decay has to be faster than any polynomial, and hence ``super-polynomial''. This confirms the intuition we have from signal processing, namely that smooth functions have to be constructed from slowly oscillating basis functions: their Fourier coefficients are concentrated in the lower-order part of the spectrum.

We want to briefly note that some encoding strategies in quantum machine learning represent data as computational basis states. In these cases we cannot work with $\mathbb{R}^N$, but we can limit the data to a continuous interval in each dimension, which we chop into $d$ bins. This allows us to work with a finite-precision approximation of $\mathbb{R}^N$ expressed by products of the cyclic group, $\mathbb{Z}_d^N = \mathbb{Z}_d \times \dots \times \mathbb{Z}_d$. Figure~\ref{fig:zd_characters} shows that the real and imaginary part of its characters $\exp{(2 \pi i \frac{x k}{d})}$ are indeed ``coarse-grained'' sine and cosine functions. A similar proof to above can be used to link the ``modulus of smoothness'' of a function over this Abelian group to a decay in Fourier space.

\subsubsection{Binary features}

\begin{figure}[t]
    \centering
    \includegraphics[width=0.8\linewidth]{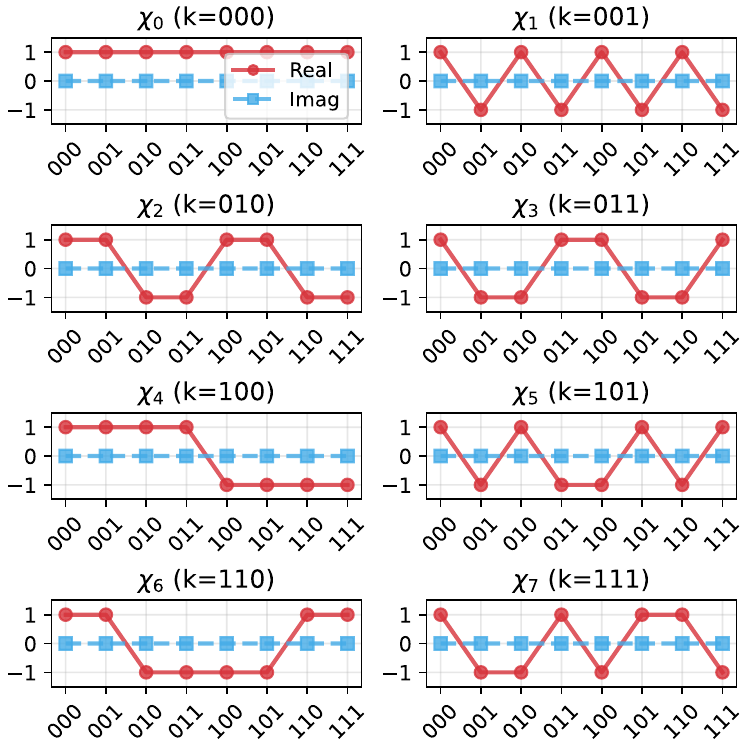}
    \caption{Fourier basis functions $\chi_k(x)$ for the boolean group $\mathbb{Z}_2^n$ with imaginary part in blue and real part in red.}
    \label{fig:z2n_characters}
\end{figure}

In Section~\ref{sec:toy_example} we saw the \textit{Walsh transform} that acts on functions on a product of cyclic groups with two elements, $\mathbb{Z}_2  \times \dots \times \mathbb{Z}_2 = \mathbb{Z}_2^n$. This group can be thought of as the set of bitstrings $x \in \{0,1\}^n$ with addition modulo $2$, and gives rise to boolean cube arithmetic. We also saw that the Fourier coefficients of a probability distribution over the group $\mathbb{Z}_2^n$ can be interpreted as expected parities of subsets of bits: The $1$s in the frequency $k \in \{0,1\}^n$ define \textit{which} bits are selected, and the Hamming weight of $k$ defines the \textit{order} of the frequency. The parity functions in Eq.~\ref{eq:parity_func} are nothing but the characters $e^{i \pi x \cdot k} = (-1)^{x \cdot k}$ (where $x \cdot k$ is the bitwise inner product modulo $2$), and the $k$ are elements from the dual group. The Fourier coefficients are then written as 
\begin{align}
    \hat{f}(k) = \frac{1}{\sqrt{2^n}} f(x) e^{i \pi x \cdot k}.
\end{align}

We will now show why the Fourier coefficients of a probability distribution over $\mathbb{Z}_2^n$ also have other statistical interpretations, namely as \textit{moments} when $f$ is a probability distribution, and \textit{interaction effects} if it is a response function (as known in the literature on experimental designs \cite{kirk2009experimental}). Together these insights show that the lower-order characters of $\mathbb{Z}_2^n$, and hence low-order frequencies of a model, have rich interpretations as the statistically simple part of a model.

\textit{Moments}. Remember that a moment of a probability distribution $p(x)$ is defined as\footnote{For ease of notation, we will not use different symbols for a random variable $X$ and its realisation $x$ as standard in the statistics literature.} 
\begin{align}
     E[x^{k}] =  \sum_x p(x) x^k.
\end{align}
For multiple binary random variables $x_1,...,x_n$, a mixed moment of order $k = k_1+ \dots + k_n$ is given as 
\begin{align}
E[x_1^{k_1}\dots x_n^{k_n}] &= \sum_{x_1,...,x_n} p(x_1,...,x_n) x^{k_1}_1 \dots x^{k_n}_n. 
\end{align}
For binary random variables, the $k_1,...,k_n$ have to be in $\{0,1\}$ as $x_i^2 = x_i^0, x_i^3 = x_i^1,$ etc. The relation between moments and Fourier coefficients is exposed if, for the moment, we move to a group that is perfectly isomorphic to $\mathbb{Z}_2$, namely the second-order cyclic group
$(\{-1,1\}, \times )$ of a ``spin variable'' under multiplication. With this change, i.e., if we define our model distribution on $\bar{x} \in \{-1, 1\}^n$ instead of $x \in \{0,1\}^n$, \textit{the Fourier coefficients of $p(\bar{x})$ are precisely its moments}. This can be seen since
\begin{align} 
    (-1)^{k x} &= \prod_{i | k_i = 1} (1-2x_i)= \prod_{i | k_i = 1} \bar{x}_i
\end{align}
where $x_i, k_i$ is the $i$'th bit of $x, k$. As the shape of the function does not change when moving from one group to its isomorphic equivalent, we can intuitively interpret Fourier coefficients of $\mathbb{Z}_2^n$ as moments. Consequently, constructing models with a decaying Fourier spectrum defines a model class that captures only low-order moments. It is important to note that this link is a special property of $\mathbb{Z}_2^n$.  More generally, the Fourier spectrum can be defined as the \textit{characteristic function} of a probability distribution, and the moments can be ``generated'' by taking partial derivatives of that function. However, we will see below that also for the symmetric group, low-order Fourier coefficients capture some kind of low-order correlation.

\textit{Interaction effects.} It turns out that a field of statistics called ``experimental design'' \cite{kirk2009experimental} offers an alternative, very intuitive interpretation for the Fourier coefficients of a boolean function $f(x)$, $x \in \{0,1\}^n$ as so-called \textit{interaction effects}, which capture how much a change in one variable affects the influence that other variables have on $f$. We will illustrate this with a small example. 

Imagine that bistrings $x$ signify combinations of drugs, modeled by features $x_1,...,x_n$ that can either be administered ($x_i = 1$) or not ($x_i=0$) for $i=1,...,N$. For now, consider every possible combination $x \in \{0,1\}^n$ of $N$ drugs was tested in a large trial, and the response $f(x)$, such as a recovery rate of patients, was measured.\footnote{In so-called \textit{fractional factorial designs} statisticians look at trials that did not include all $x$ treatment combinations, which is much closer in spirit to our ultimate goal to work with finite data.}  A crucially important question is how the first drug $x_1$ influences $f(x)$, but independent of whether or not we administered drug $x_2$. We can do this by measuring the difference between the response for $x_1 = 1$ and $x_1 = 0$, averaged over all possible values of a fixed $x_2$:
\begin{align}
L(x_1) = \big(f(10) - f(00)\big) + \big(f(11) - f(01)\big).
\end{align}
The same is true for the influence of $x_2$ on $f$:
\begin{align}
L(x_2) = \big(f(01) - f(00)\big) + \big(f(11) - f(10)\big).
\end{align}
The above can be summarised by introducing the \textit{conditional effect}
\begin{align}
L(x_1 | x_2=a) = \big(f(1a) - f(0a)\big),
\end{align}
and writing
\begin{align}
L(x_1) = L(x_1 | x_2=0) + L(x_1 | x_2=1).
\end{align}
Finally, we can ask how the drugs interact, for instance how the change in response function $f$ when changing $x_1$ \textit{differs} 
 when we change $x_2$. (This is in fact a second derivative, it is the change of the change.) The interaction can be defined by taking the difference of the conditional effects above:
\begin{align}
L(x_1 x_2) = L(x_1 | x_2=1) - L(x_1 | x_2=0).
\end{align}
Clearly, the value $L(x_1)$ corresponds to the  $\mathbb{Z}_2^n$ Fourier coefficient $\hat{f}(k=10)$ of $f$, while $L(x_2)$ corresponds to $\hat{f}(k=01)$ and $L(x_1 x_2)$ corresponds to $\hat{f}(k=11)$. This generalises to higher-order interactions: the effect $L(x_{i_1},...,x_{i_m}) = L(\{x_{i_l}\}_{l=1}^m)$ that the interaction of a subset of drugs $x_{i_1},...,x_{i_m}$ has on $f$ is given by the Fourier coefficient $f(k)$, where $k$ has ones in positions $i_1,...,i_m$ and zeros otherwise. Overall, there seems to be a deep connection between groups, spectral analysis and experimental design \cite{bailey2015spectral}, which gives the Fourier spectrum of a model a tangible interpretation. 

\subsubsection{Permutations} \label{sec:interpretation_sn}

To venture into the realm of non-Abelian groups, we want to summarise the intuition that Persi Diaconis built for the notion of simplicity contained in the Fourier spectrum of functions over the symmetric group $\mathbb{S}_n$ \cite{diaconis1988group} (see also \cite{huang2007efficient} for a gentle introduction). Loosely speaking, the Fourier coefficients are related to expected patterns such as ``\textit{objects $A D E$ are in positions $2, 4, 5$, while objects $BC$ are in positions $1, 3$},'' with more complex patterns corresponding to higher-order coefficients.

The symmetric group consists of permutations over $n$ objects, which are maps $\pi$ from a set of $n$ elements to itself. For example, for the set of $n=4$ elements $\{1,2,3,4\}$, a possible permutation is the map 
\begin{align}
    \pi(1) &= 2,\\
    \pi(2) &= 4,\\
    \pi(3) &= 1,\\
    \pi(4) &= 3.
\end{align}
The map $\pi$ can also be written as the tuple $(2, 4, 1, 3)$, in one-line notation, which is interpreted relative to the base order $(1,2,3,4)$. Permutations are combined by chaining these maps. The symmetric group is a fundamental construction in representation theory. The significance of this group in machine learning is twofold. First, numerous datasets can naturally be expressed as permutations, most notably as rankings in information retrieval (search engines), preference ordering (voting, product choices) or object tracking (computer vision and identity management). Second, as mentioned previously, permutations can also \textit{act} on the constituents of a data point, as in the above example on feature permutation. 

Without going into the vast details of the representation theory of $\mathbb{S}_n$, we note that the irreducible representations of the symmetric group are labeled by integer partitions of $n$, defined as tuples of positive integers $(\lambda_1,\lambda_2,\dots,\lambda_d)$ such that $\sum_{i}^d\lambda_i = n$ and $\lambda_1\geq \lambda_2\geq\dots\geq \lambda_d$. A natural partial order for the irreps—or the frequencies—is established by comparing the partial sums $\sum_{l=1}^{i}\lambda_l$ of two such partitions: a frequency is of higher order (dominates) if, for all possible $i$, this sum is greater than or equal to the corresponding sum of the other partition (see, for example, \cite{huang2009fourier} Def. 29). This is formally known as the \textit{dominance order}.

The entries of the matrix-valued Fourier coefficients of a function $f(\pi)$ over $\mathbb{S}_n$, are projections of $f$ onto the normalized matrix elements $\sigma(\pi)_{\lambda,ij}$ of the irreducible representations, which by Peter-Weyl theorem form an orthogonal basis for the space of all scalar functions over the group, $L^2(\mathbb{S}_n)$,
\begin{align}
\hat{f}(\lambda)_{i,j} = \sum_{\pi} f(\pi) \sigma(\pi)_{\lambda,ij}.
\end{align}
Note that we briefly switch to addressing an irrep by an index $\lambda$ to make the frequencies notationally more explicit. The elements $\sigma(\pi)_{\lambda,ij}$ form a basis for the invariant subspace (modules) $V_\lambda$ corresponding to the irreps of $\mathbb{S}_n$. Unfortunately, it is quite difficult to make intuitive sense of how such Fourier coefficients relate to the complexity of $f(\pi)$, or, in the case of probabilistic models $p(\pi)$, to low- or high-order correlations. 

Diaconis' insight was that we \textit{can} interpret the Fourier spectrum by looking instead at projections of $f(\pi)$ onto the basis vectors $e_t$ of some other non-orthogonal subspaces $M_{\lambda}$ that we call the ``marginal subspaces''. These subspaces are indexed by the same integer partitions as the irreps of $\mathbb{S}_n$, and are, technically called Young permutation modules. The invariant subspaces $V_\lambda$ are the block-diagonalisation of these permutation modules. The crux is that the basis functions of $M_\lambda$ are simply indicator functions over specific permutation patterns, defined by tabloids (specific mappings of subsets). For example, the marginal probability of a distribution $p(\pi)$ evaluating the specific pattern $S \to T$, where $\substack{\{2\} \to \{4\} \\ \{1,3,4\} \to \{1,2, 3\}}$, is given by the projection onto the corresponding indicator function:
\begin{equation}
P_{\substack{\{2\} \to \{4\} \\ \{1,3,4\} \to \{1,2, 3\}}} = \sum_{\pi} p(\pi) \mathbb{I}[(\cdot,\cdot,\cdot,2)]
\end{equation}
where the indicator function inside the sum is $1$ on all permutations where item $2$ is moved to position $4$, while $1, 3, 4$ are permuted into position $1, 2, 3$ (see Figure~\ref{fig:m_sn}). Projections onto these basis functions can be interpreted as how much weight a function has with regard to this specific structural pattern (see Figure~\ref{fig:m_sn}). With this, it becomes meaningful to associate ``lower-order'' subspaces with patterns that involve only a few elements.

\begin{figure}
    \centering
    \includegraphics[width=0.65\linewidth]{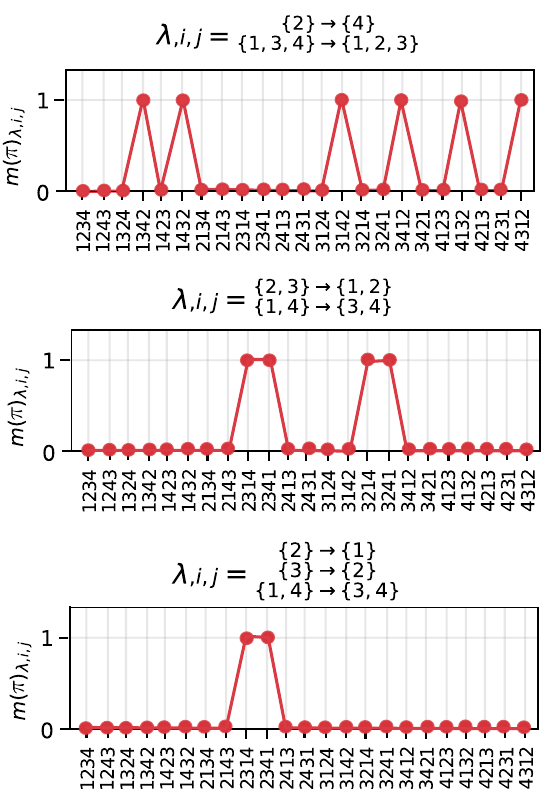}
    \caption{Examples for basis functions $\hat{m}({\lambda})_{i,j}$ of the ``marginal'' subspaces, which are indicator functions that are $1$ for certain permutation patterns. Shown are basis functions from the three marginal subspaces corresponding to frequencies indexed by the integer partitions $(3,1)$, $(2,2)$ and $(2,1,1)$. It is visible that the highest order (bottom) has support over fewer permutations, than the lowest order (top). The marginal subspaces lend interpretation to Fourier coefficients, which define different subspaces that are ``pure'' versions of those spanned by the marginal basis.}
    \label{fig:m_sn}
\end{figure}

Formally, the relationship between these statistically easy-to-interpret marginal subspaces $M_{\lambda}$ and the invariant subspaces $V_\mu$ relevant in spectral analysis (Specht modules) is given by their decomposition via Kostka numbers $K_{\mu\nu}$:
\begin{equation}
M_\lambda \cong\bigoplus_{\mu\trianglerighteq\lambda}K_{\mu\lambda}V_\mu,
\end{equation}
where $\trianglerighteq$ denotes the dominance order of the irreps. Intuitively, this means that $M_{\lambda}$ is composed of the invariant subspace $V_{\lambda}$ along with all higher-order invariant subspaces. We can therefore interpret $V_{\lambda}$ as the part of $M_{\lambda}$ that gets added when we move from a lower-order representation to ${\lambda}$. Consequently, the Fourier subspaces $V_{\lambda}$ are the ``pure'' ${\lambda}$-th order versions of the marginal subspaces $M_{\lambda}$. This allows us to understand the Fourier spectrum of the symmetric group as the ``pure'' part of statistical averages over certain permutation patterns. 

While the precise theory behind this notion is rather technical, there are some surprisingly tangible applications of this intuition for identity management \cite{huang2009fourier}, card games \cite{diaconis1988group}, the analysis of election data \cite{diaconis19891987}, cliques in the decision making of the supreme court~\cite{lawson2006spectral}, as well as for experimental designs~\cite{bailey2015spectral}. We will review a recent paper co-authored by some of us in Section~\ref{sec:symmetric_modelling} that shows how quantum computers could, in principle, help to unlock some of these methods by moving efficiently between the Fourier and direct space of functions over the symmetric group~\cite{belis2026}.

\subsection{Fourier methods in machine learning}
The machine learning literature features little \textit{explicit} mention of Fourier methods besides a few exceptions such as: Random Fourier Features \cite{rahimi2007random}, the use of the Fast Fourier Transform to speed up certain transformer layers \cite{lee2022fnet}, or specialized architectures for domains such as time-series data~\cite{yi2025FFTsurvey} or partial differential equations~\cite{li2021fno}. Kernels and convolution, on the other hand, are ubiquitous concepts that have warranted their own textbooks. Their utilization range from neural tangent kernels explaining the mechanism of deep learning \cite{jacot2018neural} and generalised convolution in geometric deep learning \cite{bronstein2021geometric}, to traditional kernel methods like support vector machines and Gaussian processes \cite{kernel, steinwart2008support, wilson2013gaussian}. Convolution and kernels, however, are deeply related to spectral methods, a connection we will illustrate in the following. We will use the general language of locally compact Abelian groups, with $g \in G$ as the elements in direct space and $k \in \hat{G}$ as the frequencies, but this can be readily translated to the more familiar case in machine learning of $G = (\mathbb{R}^N, +)$, and sometimes also to locally compact non-Abelian groups. 

\subsubsection{Stationary kernels are filters in Fourier space}
First, let us define what a kernel is.
\begin{definition}
    Let $G$ be a locally compact group. A \textit{kernel} is a symmetric, positive definite function $\kappa: G \times G \to \mathbb{C}$. 
\end{definition}
If a kernel only depends on the relative value $g g'$, it is known as a stationary kernel.
\begin{definition}
    A \textit{left stationary kernel} is a kernel that is invariant with respect to the left group action, i.e., 
    \begin{align}
        \kappa(hg, hg') = \kappa(g, g') \;\;  \forall g, g', h \in G,
    \end{align} 
    with an equivalent definition of right stationary kernels. 
\end{definition}
Note that a stationary kernel can be written as a function of only one element, 
    \begin{align}
        \kappa(g, g') = \phi(g g'^{-1}).
    \end{align}
Most of the widely used kernels are stationary, such as the Gaussian or Laplacian kernel. 

The crucial property of stationary kernels is that convolving some function $h(x)$ with such a kernel multiplies each Fourier coefficient $\hat{h}(k)$ of $h$ with the kernel's Fourier coefficient $\hat{\phi}(k)$:
\begin{align}
    (h * \phi)(g) = \sum_{g'\in G} h(g) \phi(g g'^{-1}) = \mathcal{F}^{-1} \{ \hat{h}(k)\hat{\phi}(k)\}. 
\end{align}
The convolution hence implements a \textit{spectral filter}, and the Fourier spectrum of the stationary kernel controls the properties of the resulting function $h * \kappa$.\footnote{Note that another way of stating this is that the eigenfunctions of the integral operator for stationary kernels over the uniform measure applied to some function $f$ are the Fourier basis functions, and the eigenvalues are $f$'s Fourier coefficients.} 

Many machine learning models rely on convolution, although the meaning of the abstract function $h(x)$ varies widely between different examples. 
For example, \textit{Support Vector Machines} \cite{steinwart2008support} can be shown to use the model class of linear combinations of kernels
\begin{align}
    f(x) = \sum_{x' \in \mathcal{X}} \kappa(x, x') a(x'),
\end{align}
which can be understood as the convolution of the kernel with some weight function $a(x)$ (which is only non-zero on the ``support vectors''). 

The \textit{Maximum Mean Discrepancy} \cite{gretton2012kernel} is a kernel-based distance measure for two probability distributions $p(g)$ and $q(g)$ used for generative learning and can be written as 
\begin{align}   
\text{MMD}^2(p,q) &= \sum_{g, g'} \Delta(g') \Delta(g) \kappa(g, g'),
\end{align}
with $\Delta(g) = p(g) - q(g)$. Using Plancherel's theorem, this expression can be shown to translate into $\sum_k |\hat{p}(k) - \hat{q}(k)|^2 \hat{\phi}(k)$, which means that the (squared) MMD measures the distance between the Fourier coefficients $\hat{p}(k)$, $\hat{q}(k)$ of the distributions, weighted by the Fourier coefficients of the kernel.

\textit{Convolutional neural networks} apply a convolution on a (tractable) function $I(x, y)$ that can be interpreted as a pixel value for each coordinate $x, y \in [0,\dots,d]$ of a $d\times d$-dimensional image, or of a layer after processing an image\footnote{Note that in practice, this equation gets implemented as the cross-correlation $\sum_{s,t} I(s+x, t+y) \kappa(x, y)$.},
\begin{align}
    I'(x, y) = \sum_{s,t} I(s, t) \kappa(x-s, y-t).
\end{align}
\textit{Geometric deep learning} \cite{bronstein2021geometric} essentially generalises this to other domains such as graphs. It can be understood as designing spectral filters for the inputs (understood as functions over some group), rather than for a function over group-structured inputs---a feat that is much more computationally tractable. While this is not the focus of this paper, geometric deep learning provides a lot of adjacent evidence that group spectral properties are useful to learn from data.

Note in all of the above, a common choice of a kernel for $G = \mathbb{R}^N$ is the stationary \textit{Gaussian} or \textit{radial basis kernel}:
\begin{align}
    \kappa(x, x') =  \exp\left(-\frac{\|x-x'\|^2}{2\sigma^2}\right).
\end{align}
A Gaussian function has a special Fourier spectrum, which is also a Gaussian, and hence has support over all inputs, but a strong decay that suppresses higher-order coefficients. This imposes a widely useful (albeit ``hard'', or unlearnable \cite{wilson2013gaussian}) regularisation on the model class, which biases them towards smooth, or simple models. 

Unsurprisingly, machine learning research has seen many attempts at learning the kernel, which for stationary kernels learns a bias in Fourier space. The first attempts were to combine multiple static kernels, and learning the weights that organise the influence of each of these distance measures \cite{gonen2011multiple}. Notably, Wilson et al. \cite{wilson2013gaussian} propose to use a kernel whose Fourier spectrum is a normalised mixture of a finite number of Gaussian distributions with trainable means and deviations, and thereby designs the spectral bias directly. After deep learning became popular, kernel learning moved to the idea of using deep neural networks to extract the feature vectors fed into a kernel method \cite{wilson2016deep, liu2020learning}. The quadratic dependence on the data was typically circumvented with sampling methods \cite{yang2015carte}. While still in use for domains that rely on the uncertainty quantification that kernel methods provide, learning the parameters of the kernel method together with the parameters of the neural network turned out to be difficult and issues such as ``feature collapse'' and overfitting are frequently observed \cite{ober2019benchmarking, ober2021promises}.

\subsubsection{The spectral bias of deep learning}

The \textit{spectral bias} describes the tendency of deep neural networks to learn low-frequency components of a target function significantly faster than high-frequency components. The phenomenon was first observed concurrently by \citep{rahaman2019spectral} and \cite{xu2019frequency} (who used the term ``F-Principle''). \citet{rahaman2019spectral} use an analytic derivation that allowed the efficient estimation of Fourier coefficients for neural networks with ReLU activations, and show empirically that the Fourier coefficients of the model match those of the ground truth from low to high order during training. At the same time, they find that lower-order coefficients are more robust against perturbing the \textit{parameters} of the neural network, and that adding high-frequency noise to data does not disturb the net's performance. \citet{xu2019frequency} on the other hand compare averages of the low and high-order parts of the training error's Fourier spectrum, with similar conclusions. The computational method to make statements about the Fourier transform of high-dimensional functions was later refined by \citet{kiessling2022computable}, who use a sinc-kernel to extract the averages by convolution---once more making Fourier space computationally accessible with kernels.

But not only the empirical tools, also later theoretical explanation of the spectral bias involves kernels. A range of studies \cite{cao2019towards, basri2020frequency} have contributed evidence that the Neural Tangent Kernel (NTK) \cite{jacot2018neural}---a kernel that famously describes the training dynamics of deep neural networks, and one of the major theoretical tools in the analysis of deep learning---carries this spectral bias for some neural network architectures. More precisely, it has been shown that the $k$'th eigenvalue of the NTK, $\lambda_k$, controls how fast a projection of the training error $\epsilon$ onto the corresponding eigenfunction converges \cite{arora2019exact},
\begin{align}
    \epsilon_k(t) = \epsilon_k(0) e^{-\eta \lambda_k t},
\end{align}
where $\eta$ is the learning rate. Patterns associated with large eigenvalues are therefore learned quickly, while patterns associated with small eigenvalues are learned slowly. It follows that if the NTK is stationary and the data distribution uniform (for example, when the input data is uniformly distributed over the unit sphere \cite{cao2019towards}), the eigenfunctions of the NTK are the Fourier basis functions, and the eigenvalues the network's Fourier coefficients. This provides an analytical explanation for the spectral bias for certain situations.

In summary, while spectral methods seem to lack explicit prominence in mainstream machine learning research, the widespread use of kernels and convolution shows that many established ML techniques fundamentally rely on manipulating---albeit implicitly---the Fourier spectrum of a model or distance measure.  

\section{Natural for quantum computers}\label{sec:natural_quantum}

An important---albeit, as we will discuss later in this section, not the only---argument that quantum computers could be a natural hardware for spectral methods in machine learning rests on the ability of quantum computers to efficiently implement Fourier transforms on the amplitudes of a quantum state. Given a quantum state $\ket{\psi}$, we can interpret the amplitudes $\psi(x) = \langle x | \psi \rangle$ in computational basis as a complex-valued function over $x$. The \textit{Quantum Fourier Transform} (QFT) prepares a quantum state whose amplitudes are given by the discrete Fourier transform of $\psi(x)$, 
\begin{equation}
    \sum_x \psi(x) \ket{x} \to \sum_k 
    \hat{\psi}(k) \ket{k}.
\end{equation}

But how is $\ket{x}$ associated with a group? And when are efficient quantum Fourier transforms known to exist? In this section we will put these questions onto more solid footing, which requires understanding the Fourier transform as a basis change in a vector space. We will also give explicit formulas of how the Fourier coefficients of amplitudes relate to the Fourier coefficients of standard supervised and generative quantum models built from these quantum states. Lastly, we will look at situations where quantum models can design Fourier spectra without making use of the QFT.

\subsection{Connection to quantum Hilbert spaces}\label{sec:fourier_quantum}

We can represent discrete function $f:G\to \mathbb{C}$ as a vector $v \in \mathbb{C}^{|G|}$ of function values
\begin{align}
    f \leftrightarrow \begin{bmatrix}
           f(g_1) \\
           f(g_2) \\
           \vdots \\
           f(g_{|G|}).
\end{bmatrix}
\end{align}
The Fourier transform can be understood as the matrix that changes from the standard basis 
\begin{align}
    e_{g_1} = \begin{bmatrix}1 \\ 0 \\ \vdots \\ 0 \end{bmatrix}, \dots, e_{g_{|G|}} = \begin{bmatrix}0 \\ 0 \\ \vdots \\ 1 \end{bmatrix}
\end{align} 
to the Fourier basis of character-vectors
\begin{equation}
    \left\{ \begin{bmatrix}
           \chi_k(g_1) \\
           \chi_k(g_2) \\
           \vdots \\
           \chi_k(g_{|G|}) 
    \end{bmatrix}\right\}_{k \in \hat{G}}
\end{equation}
in the Abelian case, and matrix elements of the irreducible representations
\begin{equation}
    \left\{ \begin{bmatrix}
           \sigma_{ij}(g_1) \\
           \sigma_{ij}(g_2) \\
           \vdots \\
           \sigma_{ij}(g_{|G|}) 
    \end{bmatrix} \right\}_{\sigma \in \mathcal{R}, ij \in \{1,\dots, d_{\sigma}\}}
\end{equation}
in the non-Abelian case. 

To move to quantum states, we have to associate the computational basis with the standard basis vectors 
\begin{align} \ket{g} \leftrightarrow e_{g}. 
\end{align}
This allows us to interpret the state in computational basis as the vector representation of a function $\psi(g)$ on the group,
\begin{align}
    \sum_{g \in G} \psi(g) \ket{g}.
\end{align}
The QFT is then a unitary operator which changes into the Fourier basis, and is given by
\begin{align}
    F = \frac{1}{\sqrt{|G|}} \sum_{g \in G} \sum_{\chi_k \in \hat{G}} \chi_k(g) |k \rangle \langle g|
\end{align}
for the Abelian case, and 
\begin{align}
    F = \sum_{x \in G} \sum_{\sigma \in \mathcal{R}} \sqrt{\frac{d_{\sigma}}{|G|}} \sum_{i,j=1}^{d_{\sigma}} \sigma(g)_{ij} |\sigma, i, j \rangle \langle g|,
\end{align}
for the non-Abelian case (see also \cite{Childs_2010}).

The view of the Fourier transform as a basis change has a deeper root in representation theory that we want to briefly allude to. Technically, associating the computational basis with the standard basis for each group element means associating the Hilbert space with the \textit{regular representation} of the group.  Remember that a representation is a map $R:G\to GL(V)$, which ``represents'' the group as a linear transformation on some vector space. The vector space $V$ of the regular representation is spanned by basis states $\{e_{g_i}\}$ (which means that $V$ has dimension $|G|$). The \textit{right regular representation} is then defined as the map with the property $R(h) e_g = e_{hg}$ with $h, g \in G$ (while the \textit{left regular representation} fulfills $R(h) e_g = e_{gh}$). This allows us to formulate an alternative definition of the group Fourier transform:
\begin{definition}
The Fourier transform is a basis change $F$ that block-diagonalises the (left and right) regular representation of $G$.
\end{definition}
In Section~\ref{sec:resource} we will see that if quantum states are associated with other representations than the regular one fixed by interpreting a computational basis state with a group element, we can perform a spectral analysis that is similar to Fourier analysis. However, the basis change is given by other transforms (such as the quantum Schur transform \cite{bacon2006efficient, krovi2019efficient, burchardt2025high}).

\subsection{When do efficient QFTs exist?}

An efficient QFT can be decomposed into $O(\text{poly}(n))$ elementary gates. While the most prominent version is the QFT over the cyclic group $\mathbb{Z}_{2^n}$, which facilitates Shor’s factoring algorithm with a gate complexity of $O(n^2)$, the QFT is known to be efficient for all finite Abelian groups. This generalization is made possible by the fundamental theorem of finitely generated Abelian groups, which states that any finite Abelian group is isomorphic to a direct product of cyclic groups. Since the QFT of a direct product group $G \times G'$ can be implemented as the tensor product of the individual transforms ($QFT_G \otimes QFT_{G'}$), the ability to efficiently transform cyclic factors implies efficiency for the entire Abelian class\footnote{See Section 6.2 in \cite{childs2017lecture} on the problem of \textit{finding} the isomorphic cyclic group.}. 

Beyond the Abelian case, efficient QFT algorithms have been established for certain non-Abelian families, such as the symmetric group $\mathbb{S}_n$, metacyclic groups, the wreath product of poly-sized group, and metabelian groups \cite{beals1997quantum, puschel1999fast, hoyer1997efficient, moore2006generic}. Many of these QFTs rely on the mechanism behind Cooley-Tucker's Fast Fourier Transform (FFT), which uses a ``subgroup tower'' to decompose a Fourier transform on a group into those over subgroups. Quantum computers can elegantly parallelise this divide-and-conquer strategy \cite{moore2006generic}, turning the polynomial speedup alluded to by the term ``Fast'' into an exponential (and sometimes super-exponential) one. As a result, an efficient QFT is known to exist for every group that admits an FFT.\footnote{See the PennyLane tutorial \cite{MariaSchuld2025} for an intuitive explanation.}

\subsection{QFTs and generative quantum models} \label{sec:qft_generative_qml}

Quantum Fourier transforms allow us to transform the \textit{amplitudes} of quantum states. But how does this relate to the spectrum of a machine learning model derived from the state? 

Let us explain this for the example of \textit{quantum circuit Born machines} \cite{liu2018differentiable}, which prepare a (in general, trainable) quantum state $\ket{\psi_{\theta}}$ and define an implicit probabilistic model \cite{mohamed2016learning} as the measurement probability in the computational basis,
\begin{align}
    p_{\theta}(x) = |\langle x | \psi_{\theta}\rangle|^2.
\end{align}
According to the discussion above, the bitstring $x \in \{0,1\}^n$ is interpreted as a group element $g\in \mathbb{Z}_2^n$.

For Abelian groups, the relation between the Fourier coefficients of the model and those of the amplitudes is given by a simple formula:
\begin{theorem}
Let $p(x)=|\langle x | \psi_{\theta}\rangle|^2$ be the measurement distribution of a (trainable) quantum state, with Fourier transform $\hat{p}(k)$. Let  $\hat{\psi}(k)$ be the Fourier coefficients of the amplitudes $\psi(x) = \langle x | \psi_{\theta}\rangle$. Then 
\begin{align}
    \hat{p}(k) = \frac{1}{\sqrt{2^n}} \sum_{s }\hat{\psi}(s) \hat{\psi}^* (s + k)   
\end{align}
\end{theorem}
Note that we did not specify the Abelian group of the Fourier transform, as the following proof works whether we, for example, interpret the bitstrings $x \in \{0,1\}^n$ as elements from $\mathbb{Z}_2^n$ (i.e., as binary features) or as elements from $\mathbb{Z}_d^N$ (i.e., integer-valued, or coarse-grained continuous-valued features).
\begin{proof}
The proof is straight forward; we express the amplitudes in the Born rule by their Fourier decomposition and use two well-known properties of characters, $\chi_k(x)\chi_s(x) = \chi_{k + s}(x) $ and the character orthogonality theorem $\sum_x \chi_{k}(x) \chi^*_s(x) = |G| \delta_{s, k} $:
\begin{align}
    \hat{p}(k) &= \frac{1}{\sqrt{2^n}} \sum_x |\psi (x)|^2 \chi_k(x)\\
    &= \frac{1}{2^n} \frac{1}{\sqrt{2^n}} \sum_x \sum_{s} \hat{\psi}(s) \chi_{s}(x) \sum_{t} \hat{\psi}^*(t) \chi^*_{t}(x)  \chi_k(x)\\
    & =  \frac{1}{2^n} \frac{1}{\sqrt{2^n}} \sum_{s t} \hat{\psi}(s) \hat{\psi}^*(t) \sum_x \chi^*_{s - t}(x) \chi_{k}(x)\\
    &=  \frac{1}{\sqrt{2^n}} \sum_{s} \hat{\psi}(s) \hat{\psi}^* (s-k). 
\end{align}
\end{proof}
An analogous statement for non-Abelian groups can be found in~\cite{belis2026}. Note that for $Z_2^n$, $s-k = s+k$. 

We want to briefly mention another example of how QFTs could help with the design of quantum models. One of the most important use cases for QFTs is to solve \textit{hidden subgroup problems} \cite{Childs_2010}, of which Shor's algorithm is the most well-known. While hidden subgroup problems seem rather artificial for a discipline as applied as machine learning, recently it was shown that it can be used to find partitions of qubits into unentangled subsets \cite{bouland2024state}, an algorithm which can elegantly be translated into heuristics that gives us access to the (un-)entanglement structure of a quantum state \cite{simidzija2026approximate}. If such a quantum state represents a generative machine learning model, quantum algorithms for hidden subgroup problems based on the QFT could provide unique tools to query and manipulate the independence structure of variables represented by the qubits. Another avenue in this direction is the exploration of learning hidden subgroups from data \cite{wakeham2024inference} (rather than computing it from an oracle), although it is yet unclear how this  problem relates to real-world applications.

\subsection{Spectral methods beyond the QFT}\label{sec:fourier_qnns}

The situation of spectral biases looks very different for ``quantum neural networks'' \cite{farhi2018classification, schuld2019evaluating}. Here the computational basis that the QFT acts on is not associated with the data any more, which is in general continuous-valued. Instead, the data is typically encoded into the state using gates $e^{i x_i H}$ that encode the elements of the vector via an evolution of the Hermitian operator $H$. Note that as $e^{i x_i H} = e^{i (x_i+2\pi) H}$, the data domain is effectively the group $\mathbb{R}/\mathbb{Z}$. The model function $f(x)$ is defined as the expectation of some observable $O$ with respect to the data-dependent state $\ket{\psi_x}$:
\begin{align}
    f(x) = \langle \psi_x | O | \psi_x \rangle.
\end{align}
As a result, the quantum Fourier transform is not a Fourier transform on the data space any more. However, the quantum model still has special properties with respect to the ``classical'' Fourier transform over $\mathbb{R}/\mathbb{Z}$, as we will explain in  Section~\ref{sec:qnn_bias}. This makes quantum neural networks candidates to design spectral biases \textit{without} applying a quantum Fourier transform, simply by the way that data is encoded.

\section{Spectral methods in QML}\label{sec:spectral_qml}

Following the above collection of arguments for why spectral methods are crucial for machine learning and natural for quantum computers, we want to conclude by reviewing a few existing areas in quantum machine learning research that already recognise the potential of spectral methods. The material collected in the previous sections suggests that these compile just the beginning of a much deeper connection that is waiting to be uncovered.

\subsection{The spectral bias of QNNs}\label{sec:qnn_bias}

Quantum Neural Networks---also referred to as Variational Quantum Circuits---are parametrised quantum circuits used as supervised machine learning models \cite{farhi2018classification, schuld2020circuit}. A subset of circuit parameters is used to encode an input $x \in \mathbb{R}^N$, while the remaining parameters $\theta$ are trained by gradient descent methods, typically using parameter-shift rules \cite{mitarai2018quantum, schuld2019evaluating}. An expected observable, such as the average measurement result for a designated qubit, is interpreted as the value $f_{\theta}(x)$ of the model. A growing body of literature argues that Quantum Neural Networks have a spectral bias \cite{schuld2021effect, lu2026unified, duffy2025spectral, xu2025spectral} which can be manipulated for specific learning tasks \cite{jaderberg2024let}. This includes a ``hard'' spectral bias stemming from the embedding of classical data and a potential ``soft'' spectral bias that regularises the underlying model class \cite{schuld2021effect}, as well as a possible spectral bias with respect to the learning dynamics similar to the one observed in classical neural networks \cite{lu2026unified}. 

The ``hard'' spectral bias goes back to \citet{schuld2021effect} which showed that Quantum Neural Networks can be expressed by a truncated Fourier series. A feature $x \in \mathbb{R}$ of the input vector (as well as the trainable parameters) are encoded via gates of the form $e^{i x H}$, where $H$ is a Hermitian operator that is often chosen as a single-qubit Pauli operator, which makes the gate a $X, Y$, or $Z$ rotation. Assuming that $H$ has integer-valued eigenvalues, we have $e^{i x H} = e^{i (x + 2 \pi) H}$, which means that the model $f_{\theta}(x)$ is $2 \pi$-periodic. Alternatively, the data can be thought of as taken from the $N$-dimensional unit circle, or the group $\mathbb{R}/\mathbb{Z}$. The natural Fourier decomposition is therefore the Fourier series 
\begin{align}
    f_{\theta}(x) = \frac{1}{\sqrt{2\pi}}\sum_{k \in \Omega} \hat{f}_{\theta}(k) e^{-2 \pi i x \cdot k}  
\end{align}
where the elements of $k$ can take any integer value ($k = {-\infty, \dots, \infty})$ and $x \cdot k$ is the standard inner product of two vectors. It turns out \cite{schuld2021effect} that the eigenvalues $\{\lambda_1, \dots, \lambda_d \}$ of $H$ put limitations on $\Omega$, and therefore bandlimit the Fourier spectrum: the spectrum only contains frequencies 
\begin{align} 
  \Omega = \{k | \Lambda_{\mathbf{i}} - \Lambda_{\mathbf{j}}= k\}.
\end{align} 
composed of combinations of eigenvalues as 
\begin{align}
    \Lambda_{\mathbf{i}} &= \left(\lambda_{i_1} + \dots + \lambda_{i_d}\right)\\
    \Lambda_{\mathbf{j}} &= \left(\lambda_{j_1} + \dots + \lambda_{j_d}\right),
\end{align}
where $\mathbf{i} = (i_1, ...,i_d)$, $\mathbf{j} = (j_1, ...,j_d)$ and $i_l, j_l = 1,...d$.
In other words: typical Quantum Neural Network architectures are a bandlimited function class. 

Within the allowed frequency band, Quantum Neural Networks  can also have a soft spectral bias. Essentially, this can be seen by considering the combinatorial expression linking the Fourier coefficients to parameters $a_{\mathbf{j}, \mathbf{k}}$ that depend on the choice of gates and measurements in the circuit,
\begin{align}
    \hat{f}_{\theta}(k) = \sum_{\mathbf{i}, \mathbf{j}| \Lambda_{\mathbf{i}} - \Lambda_{\mathbf{j}} = k} a_{\mathbf{j}, \mathbf{k}}.
\end{align}
The smaller $k$, the more combinations of eigenvalues $\Lambda_{\mathbf{i}} - \Lambda_{\mathbf{j}}$ add up to $k$, and the more terms in the sum. This has led to studies exploring how the embedding gates' eigenvalue spectrum shapes the Fourier spectrum of Quantum Neural Networks \cite{mhiri2025constrained}, and to claims of a spectral bias based on redundancy \cite{duffy2025spectral}.  

In addition, there is an ongoing debate on whether quantum neural networks have a spectral bias in their learning dynamics, similar to classical neural networks. \citet{lu2026unified} provide a general framework of conditions for which the high-order part of the training error is guaranteed to be larger than the lower-order part during gradient descent, and argue that these conditions are fulfilled by both classical neural networks and quantum neural networks, even if data encoding is not facilitated by gates of the form $e^{i x_i H}$. An opposing view is provided by \cite{xu2025spectral}, whose analysis of the neural tangent kernel of quantum neural networks does not show a frequency-wise decoupling of the training error in Fourier space. 

An interesting aspect of the spectral bias debate is that a bandlimited or decaying Fourier spectrum can make classical simulation of quantum neural networks feasible, which can both boost and harm the claim that they require quantum hardware to be implemented. On the one hand, for Quantum Neural Networks to benefit from quantum hardware (and justify investments into this technology), they need to be difficult to classically simulate. On the other hand, as first-generation quantum computers will likely be small and slow, it is attractive to outsource as much as possible to classical simulation. For example, \cite{schreiber2023classical} shows how quantum neural networks need to be trained on quantum computers due to the exponential growth of trainable Fourier coefficients, but could then be deployed by ``classical surrogate models'' that implement the trained truncated Fourier series on classical hardware. \citet{li2025dual} argue that a certain type of bandlimited simulation, called \textit{Pauli Path Propagation} \cite{rudolph2025pauli}, can be used to warm-start training on quantum computers (we will discuss below how the Pauli Path Propagation algorithm can be understood to be working in a generalized Fourier basis).

Ultimately, only empirical experiments will determine whether the region of Fourier space that is classically intractable---yet quantumly accessible---plays a critical role in these models. However, we consider this highly likely, given the established importance of learning intermediate frequencies for generalizing beyond mere interpolation in deep learning.

\subsection{Probabilistic modeling over the symmetric group} \label{sec:symmetric_modelling}

In Section~\ref{sec:interpretation_sn} we motivated that the Fourier spectrum of a fundamental, but non-trivial group, the symmetric group, contains important statistical information for the design of machine learning models for permutation data. For probabilistic models, Fourier coefficients correspond to expectations of the irreducible representations of the group, and low-order Fourier coefficients reflect simple correlation patterns such as \textit{object $i$ is in position $j$}, while higher-order coefficients capture high-order correlations of the type \textit{objects $i,j,k$ are in positions $l,m,n$ while object $r$ is in position $q$}. This property has been exploited by Risi Kondor \cite{kondor2011non} and Jonathan Huang \cite{huang2009fourier} to build probabilistic models over permutations, and subsequently been translated to a quantum algorithm \cite{belis2026}. 

The basic idea in both classical and quantum proposals is to create a model in a recurrent, Markov-chain type structure, where starting with a canonical assignment between objects and positions, we iterate over \textit{diffusion} and \textit{conditioning} steps. Diffusion is a process that captures the chance of swapping the positions of two elements (i.e., a transposition) as time goes by, and can be seen as a random walk over the Cayley graph of permutations, whose edges connect permutations that are only a transposition apart. This process
``smears out'' the probabilities and creates uncertainty about object-position assignments. Conditioning, instead, \textit{adds} information to the model in the form of a Bayesian update. Essentially, one multiplies the model with a likelihood function that reflects information from making an observation such as \textit{item $j$ is in position $i$}, after which permutations with this pattern become more likely in the model. While this framework is a little different from the standard training of a neural network, it is a machine learning method similar to Kalman filtering, which is widely used in navigation and control \cite{auger2013industrial}. 

Importantly, the two steps of uncertainty generation and elimination are elegantly interpreted from the perspective of harmonic analysis. Diffusion is a convolution $p(\pi) * q$, in other words an equivariant map \cite{kondor2018generalization}, between the present model and a ``kernel'' such as 
\begin{equation}
    q(\pi)=
    \begin{cases}
        p\;&\pi=e\\
        (1-p)/\binom{n}{2}\; &\pi \text{ is a transposition}\\
        0\;& \text{otherwise}.
    \end{cases}
    \label{eq:diffusion_kernel}
\end{equation}
According to the convolution theorem, the convolution is a simple point-wise product in Fourier space. Furthermore, the above kernel is a \textit{class function}, which means according to Schur's lemma that the Fourier coefficient has the even simpler form of $\hat{q}_{\sigma} = c_{\sigma} \mathbb{I}$. This means that implementing diffusion is particularly elegant in Fourier space. Conditioning or a Bayesian update, in turn, is a product in direct space but a convolution in Fourier space (which mixes frequencies to some extent). 

The super-exponential $n!$ growth of the symmetric group, and consequently the scaling of the Fast Fourier Transform, makes probabilistic modeling over permutations, including the strategy above, challenging. To circumvent this, work by \citet{kondor2011non} and \citet{huang2009fourier} proposed operating in a bandlimited Fourier space, which entails retaining only the low-order Fourier coefficients. Such models can represent simple correlations in the data, and solve specific inference tasks that, likewise, only require these low-order Fourier coefficients. However, even with advanced implementations and severe bandlimiting, these simple models can only realistically deal with $n<30$ items \cite{kondor2011non}. In \citet{belis2026} the authors (including some of us) therefore set out to ask whether quantum computers could lift these restrictions. They show that, in principle, the Markov model after $t$ time steps can be implemented efficiently on a quantum computer using amplitude encoding, and provide an algorithm to prepare a state proportional to
\begin{align}
    \sum_{\pi \in \mathbb{S}_n} p^{(t)}(\pi) \ket{\pi}.
\end{align}
This approach requires that only a few conditioning steps are involved, and that the observations have non-negligible support in the current model at any given step. While this encodes the probabilistic model into the amplitudes of the quantum state, a similar process could prepare a model $\tilde{p}$ in ``Born encoding'' 
\begin{align}
    \sum_{\pi \in \mathbb{S}_n} \sqrt{\tilde{p}(\pi)} \ket{\pi}.
\end{align}
The model state could be used to sample from the distributions $|p(\pi)|^2$ or $\tilde{p}(\pi)$ to generate new data. Additional techniques can lead to algorithms that sample the most probable permutations by preparing a state with amplitudes proportional to $|p(\pi)|^m$, or marginals thereof---a task that is highly challenging classically.

While many non-trivial open questions remain, such as the feasibility of implementation on realistic hardware, or the performance of the model on real-world data compared to methods not based on harmonic analysis, this study opens up the prospect that quantum computers may be able to unlock machine learning with permutation-structured data. Such data is surprisingly common in rankings for recommendation systems and information management, or in identity management and tracking tasks.

\subsection{Resource theories and Fourier analysis} \label{sec:resource}

Lately, a form of group spectral methods have been suggested to study the \textit{resourcefulness} of quantum states \cite{bermejo2025characterizing} (see also the PennyLane demo \cite{PaoloBraccia2025}). Resource theories in quantum information ask how ``complex'' a given quantum state is with respect to a certain measure of complexity, which often translates into how difficult they are to prepare in the lab, or how difficult they are to simulate on a classical computer. Examples of resources are \textit{entanglement}, \textit{Clifford stabilizerness} and \textit{Gaussianity}. As it turns out, a very useful ``resource fingerprint'' of a quantum state is a generalised version of its \textit{power spectrum} (i.e., the absolute square values of the Fourier coefficients). ``Generalised'' here refers to the fact that we invoke the mechanism of Fourier analysis described in Section \ref{sec:fourier_quantum}, but instead of a projection of a function onto the orthogonal basis of all characters or irreps, we project only onto a few of them. Technically speaking, this means that we relax the condition that a quantum state is associated with the regular representation, which means that a computational basis state relates to each group element, as introduced in Section~\ref{sec:fourier_quantum}. Instead, we can associate the vector space of quantum states with \textit{any} representation, which defines what we consider to be a resource. Intuitively, this generalises the notion of ``smoothness'' inherent in the Fourier spectrum to other resources, and potentially opens up recipes for generalised regularisation of quantum machine learning models. 

The recipe to compute the group generalised power spectrum (or ``GFD Purities'' \cite{bermejo2025characterizing}) goes as follows:
\begin{itemize}
    \item Identify the set of ``free states'' under the resource, such as the set of entanglement-free product states.
    \item Identify a unitary representation $R:G\to GL(V)$ of some group $G$ that maps free states to free states. Here, $V$ is the vector space that describes quantum states. For example, if $V = \mathcal{L}(\mathcal{H})$ is the space of density matrices $\rho$, a representation may be given by the adjoint representation  $R: g \to R(g) \rho R(g)^{\dagger}$ of unitary operators $R(g)$ that do not entangle qubits.
    \item Find the basis change that block-diagonalises the representation to decompose it into irreps. These reveal the invariant subspaces of $V$. 
    \item Given a new state, applying the basis change will reveal the GFD Purities. 
\end{itemize}
Of course, if the representation is chosen as the regular representation and $V$ is the Hilbert space of Dirac vectors, the Purities are simply the absolute squares of the Fourier coefficients of a quantum state, and the block-diagonalisation is the Quantum Fourier Transform. For other representations, we might need other transforms (such as the quantum Schur transform \cite{burchardt2025high}), but the Purities have a very similar spectral interpretation.

When we take the resource to be entanglement, the invariant subspaces that correspond to different Purities are spanned by ``constant-order'' Pauli basis vectors, which are those that apply identities to a constant number of qubits. For example, for two qubits, the invariant subspaces that the Purities project into are spanned by $\{\mathbb{I} \times \mathbb{I}\}$, $\{\mathbb{I} \times \mathbb{X}, \mathbb{I} \times \mathbb{Y}, \mathbb{I} \times \mathbb{Z}\}$, $\{\mathbb{X} \times \mathbb{I}, \mathbb{Y} \times \mathbb{I}, \mathbb{Z} \times \mathbb{I}\}$, and $\{\mathbb{X} \times \mathbb{X},\mathbb{X} \times \mathbb{Y}, \dots, \mathbb{Z} \times \mathbb{Z}\}$. This shows how the invariant subspaces refer to subsets of qubits (similar to the frequencies of the Walsh transform). Quantum states with ``bandlimited'' GFD Purities are those that have limited entanglement, as they only have non-zero projections into subspaces that have non-trivial support on few qubits. Circuits that stay in these bandlimited subspaces are exactly the ones that the popular technique \textit{Pauli Path Propagation} simulates \cite{rudolph2025pauli, li2025dual}. Again, the Fourier spectrum, even if a much generalised notion of it, uncovers ``simplicity''. 

We want to note in passing that there is an intriguing relation between resource theories and \textit{quantum phase spaces}, which represent a quantum state as a function $f:G \to \mathbb{C}$ from a Reproducing Kernel Hilbert Space. It can be shown that canonical choices of the phase space implement a spectral filter: they emphasise the information of the quantum state in some irreps and suppress the information from others \cite{coffman2026group}. This opens up the prospect of using quantum phase spaces as regularised function spaces to learn about the properties of quantum states from measurements.

\section{Outlook}
In this article, we argued that spectral methods are a compelling candidate to answer the question of why quantum computers could help to design the next generation of machine learning methods, as first asked a decade ago \cite{wittek2014quantum}. After some scrutiny, spectral methods appear rather central to learning: The Fourier spectrum of a model contains information on properties such as smoothness and robustness, and is therefore the prime target for past and present regularisation strategies, which are at the heart of generalisation. Fourier methods, in particular their group-theoretic generalisations, are also the very foundation of quantum computing as they underlie quantum mechanics as a whole \cite{woit2017quantum}. The Quantum Fourier Transform, in particular, is a powerful subroutine to access---and hence manipulate---the Fourier spectrum or other algebraic properties \cite{bouland2024state} of quantum states that represent machine learning models. Can quantum computers unlock more ``direct'' spectral regularisation techniques? Will they enable the design of desirable biases without reverting to the unsustainable scales of heavily overparametrised neural networks? Are they the solution to yet unsuccessful attempts of training kernel methods? 

We hope that this collection of material helps to stimulate a new research direction in quantum machine learning that tries to find answers to these questions. To do so successfully requires us to fundamentally shift our perspective: from \textit{what is a model class with provable quantum speedups} to \textit{how can we use the strengths of a quantum computer to build good models?}  Moving towards the latter requires stronger roots in machine learning research, because it demands a good understanding for what makes a ``good'' model---in particular as we do not currently have access to large-scale benchmarking that allows for a purely empirical validation. Luckily, recent years of machine learning research have provided many answers to this question, which---perhaps unsurprisingly---have led us back to the principles of traditional learning theory \cite{wilson2025deep, belkin2019reconciling, bartlett2020benign}, rather than uprooting them \cite{zhang2016understanding}. What emerges is a recalibrated understanding of the interplay between flexibility and simplicity: Good models have \textit{soft simplicity biases}; they can learn anything, but have a preference towards simple functions. At the same time, to handle modern machine learning tasks, good models have to be computationally extremely efficient on the available hardware \cite{hooker2021hardware}. Neural networks often fulfill these properties, and are by now not so much a distinct model class, but the workhorse whenever a machine learning model requires a parametrised function whose gradients are quick to compute. However, neural networks come with an important caveat: somewhat paradoxically, they achieve simplicity through large architectures and big data inputs. This makes them not only resource intensive, but also unsuitable for small problems, where engineering features and biases is still crucial. One way to innovate machine learning is therefore to take the powerful mechanisms of neural networks and ``reproduce'' them at smaller scales. Spectral methods may offer a domain-independent blueprint for this, in particular if we take the structure of data more seriously, as evident in geometric deep learning \cite{bronstein2021geometric}.

Showing that quantum computers could implement soft simplicity biases while being computationally efficient is no small task. Quantum computers offer a very particular access to information via measurement. For example, generative quantum models can fundamentally not estimate likelihoods of the sampling distribution, whereas classical generative models (besides GANs) can be understood as a collection of ingenious, technical, and non-obvious tricks to do exactly that \cite{murphy2022probabilistic}. Quantum Fourier Transforms do not enable us to \textit{directly compute} the Fourier coefficients of a quantum state, we can only manipulate them within the limits of quantum algorithmic tools, or sample from the spectrum. And, while a fundamental building block for even fault-tolerant quantum machine learning algorithms, training parametrised circuits is by no means the silver bullet that neural networks represent: it is slow and expensive, and will likely require tricks to ``train classically, deploy quantum'' \cite{recio2025train}. 

The way towards performant spectral quantum machine learning models will be to understand and overcome these challenges, not by copying what has worked in classical learning, but by combining highly specific, theoretically well-motivated tools that are tailor-made for quantum computers. Good quantum machine learning models will also likely be dequantisable if the problems are simple enough, as they possess a lot of structure---which, we firmly believe, bypasses most of the current arguments regarding barren plateaus~\cite{larocca2025barren}, as those generally rely on Haar-random ensembles. However, even deep learning models can, once trained, be replaced by much simpler (compressed) architectures; still, generalisation relies on the flexibility of large scales~\cite{frankle2018lottery,wilson2025deep}. In the language of spectral biases: we expect that performance strongly depends on what happens in the medium-order region of the Fourier spectrum---in other words, at the edge of classical simulatability. If quantum models are performant, they \textit{should} lead to great classical models as well, which, once widely deployed, will eventually demand quantum hardware to push performance even further. We therefore argue that the current preoccupation with dequantisation in quantum machine learning can be an obstacle, rather than a compass, for true progress.   

Finally, we want to remark that while we have focused on Fourier analysis, the Fourier basis is only one of many representations that quantum computers can efficiently manipulate. There are promising suggestions for efficient quantum Schur \cite{burchardt2025high, krovi2019efficient, bacon2006efficient}, Chebyshev \cite{williams2023quantum}, Paldus \cite{burkat2025quantum}, wavelet \cite{fijany1998quantum}, Laplace \cite{shehata2020general}, and Hilbert \cite{jha2025quantum} transforms. In this spirit, quantum computers can be understood as physical samplers from high-dimensional distributions, capable of moving seamlessly into complex bases to extract features from data. It is therefore difficult to imagine that such machines will \textit{not} open up entirely new ways to learn.

\section{Acknowledgments}
We are grateful to Jason Pye, Juan Miguel Arrazola, Derek K. Wise, Gordon Brown, and Nicholas Shorter for useful discussions and feedback on the manuscript. No Large Language Model was used to generate contents of this manuscript.

\bibliography{references.bib}

\end{document}